\documentclass[10pt,journal,compsoc,onecolumn]{IEEEtran}
\ifCLASSOPTIONcompsoc
  \usepackage[nocompress]{cite}
\else
  \usepackage{cite}
\fi
\usepackage{epsfig,amsthm, algorithm, algorithmic}
\newcommand{\argmax}{\mathop{\rm arg~max}\limits}
\newtheorem{theorem}{Theorem}
\usepackage{subfigure,cite,graphicx,amsmath,amssymb,mathrsfs,epsf}
\usepackage[usenames]{color}
\usepackage{multirow}
\usepackage{tabularx}
\usepackage{fixltx2e}
\usepackage{epstopdf}
\usepackage{textcomp}
\usepackage{tikz}
\usepackage{ctable}
\usepackage{lipsum}
\usepackage{etex}

\usepackage{amsthm}
\usepackage{verbatim}
\usepackage{xcolor}
\usetikzlibrary{calc}
\usepackage{soul}
\usepackage{pgfplots}
\usepackage{stfloats}
\usepackage{setspace}
\usepackage{enumitem}
\begin{document}
\sloppy
\title{Interference-Aware Opportunistic Random Access in Dense IoT Networks}

\author{Huifa Lin,~\IEEEmembership{Member,~IEEE,}
        Kwang~Soon~Kim,~\IEEEmembership{Senior Member,~IEEE,}\\ and
        Won-Yong Shin,~\IEEEmembership{Senior Member,~IEEE}
\IEEEcompsocitemizethanks{\IEEEcompsocthanksitem H. Lin is with Telecommunication and Image Technology Laboratories of Sharp Corporation, Chiba-shi 261-8520, Japan.\protect\\
Email: huifa.lin.dr@ieee.org
\IEEEcompsocthanksitem K. S. Kim is with the Department of Electrical and Electronic Engineering, Yonsei University, Seoul 03722, Republic of Korea.\protect\\
Email: ks.kim@yonsei.ac.kr
\IEEEcompsocthanksitem W.-Y. Shin (corresponding author) is with the Department of Computational Science and Engineering, Yonsei University, Seoul 03722, Republic of Korea.\protect\\
E-mail: wy.shin@yonsei.ac.kr
 }}


\newtheorem{definition}{Definition}
\newtheorem{thm}{Theorem}
\newtheorem{lemma}{Lemma}
\newtheorem{example}{Example}
\newtheorem{corollary}{Corollary}
\newtheorem{proposition}{Proposition}
\newtheorem{conjecture}{Conjecture}
\newtheorem{remark}{Remark}

\newcommand{\red}[1]{{\textcolor[rgb]{1,0,0}{#1}}}

\def \diag{\operatornamewithlimits{diag}}
\def \min{\operatornamewithlimits{min}}
\def \max{\operatornamewithlimits{max}}
\def \log{\operatorname{log}}
\def \max{\operatorname{max}}
\def \rank{\operatorname{rank}}
\def \out{\operatorname{out}}
\def \exp{\operatorname{exp}}
\def \arg{\operatorname{arg}}
\def \E{\operatorname{E}}
\def \tr{\operatorname{tr}}
\def \SNR{\operatorname{SNR}}
\def \dB{\operatorname{dB}}
\def \ln{\operatorname{ln}}
\def \bmat{ \begin{bmatrix} }
\def \emat{ \end{bmatrix} }
\def \be {\begin{eqnarray}}
\def \ee {\end{eqnarray}}
\def \ben {\begin{eqnarray*}}
\def \een {\end{eqnarray*}}

\newcommand{\Pro}[1]{\mathrm{Pr}\left\{#1\right\}}
\newcommand{\LIF}[2]{\tilde{L}_{\pi_1(#1),#2}}
\newcommand{\TIL}[2]{L_{\pi_2(#1),#2}}
\newcommand{\TIF}[2]{T_{\pi_1(#1),#2}}
\newcommand{\KIF}[2]{T_{\pi_1(#1),\pi_2(#2)}}
\newcommand{\snr}{\textsf{snr}}
\newcommand{\sinr}{\textsf{sinr}}
\newcommand{\CanSB}{\mathcal{B}}
\newcommand{\CanSA}{\mathcal{A}}
\newcommand{\Norm}[1]{\left|{#1}\right|}

\IEEEtitleabstractindextext{
\begin{abstract}
It is a challenging task to design a random access protocol that achieves the optimal throughput in multi-cell \emph{random access} with decentralized transmission due to the difficulty of coordination.
In this paper, we present a decentralized \emph{interference-aware opportunistic random access (IA-ORA)} protocol that enables us to obtain the optimal \emph{throughput scaling} in an ultra-dense multi-cell random access network with one access point (AP) and a number of users.
In sharp contrast to \emph{opportunistic scheduling} for cellular multiple access where users are selected by base stations, under the IA-ORA protocol, each user opportunistically transmits with a predefined physical layer (PHY) data rate in a decentralized manner if not only the desired signal power to the serving AP is sufficiently large but also the generating interference leakage power to the other APs is sufficiently small (i.e., two threshold conditions are fulfilled).
As a main result, it is shown that the optimal aggregate throughput scaling (i.e., the MAC throughput of $\frac{1}{e}$ in a cell and the power gain) is achieved in a high signal-to-noise ratio regime if the number of per-cell users exceeds some level.
Additionally, it is numerically demonstrated via computer simulations that under practical settings, the proposed IA-ORA protocol outperforms conventional opportunistic random access protocols in terms of aggregate throughput.
\end{abstract}
}

\maketitle
\IEEEdisplaynotcompsoctitleabstractindextext

%
\IEEEpeerreviewmaketitle


\section{Introduction}
\label{sec:introduction}
Recently, the use of random access in wireless communications has received considerable attention along with the rapid development of machine-type communications (MTC) and Internet of Things (IoT) networks due to the necessity of a relatively low protocol overhead and high spectral efficiency \cite{rajandekar2015survey, Bockelmann2016}.
For several decades, a variety of random access protocols have been implemented based on ALOHA and its variants with carrier sensing \cite{Bertsekas1992}.
In MTC and IoT networks, transmission activity of vast devices tends to be irregular and unpredictable with short packets \cite{rajandekar2015survey}.
Under the assumption of such a traffic pattern, instead of state-of-the-art media access control (MAC) protocols (e.g., carrier-sense multiple access with collision avoidance (CSMA/CA)), rather simple uncoordinated random access protocols (e.g., slotted ALOHA) that incur much less protocol overheads become favorable, where protocol overheads correspond to the waiting time before users access the channel due to the channel sensing and random backoff (BO) in CSMA/CA.
However, the major problem of slotted ALOHA is its low MAC layer efficiency.
To overcome this problem, there have already been research efforts on improving the MAC throughput by introducing coded slotted ALOHA \cite{paolini2015coded} and cooperative slotted ALOHA for multi-cell random access environments \cite{jakovetic2015cooperative}.

Furthermore, along with the large increase in mobile users and their data packets in IoT networks, there has been a trend of network densification in ultra-dense networks (UDNs) with sufficiently many users \cite{chen2017machine}, e.g., a network environment where a cell in which one access point (AP) covers contains a plenty of users. Such network configuration can be observed in a variety of future MTC and IoT applications.
It is thus crucial to grasp the nature of random access dense IoT networks that are composed of multiple cells sharing the same spectrum, so called \emph{multi-cell random access} networks.\footnote{Here, we use the term ``cell'' to represent the domain of both an AP and its associated users.}
In such networks, in addition to the intra-cell collision (simultaneous transmission from multiple users in the same cell), transmission with no coordinated scheduling among APs causes interference to other-cell APs, which may result in a failure of packet decoding at the receivers.
Hence, inter-cell interference should also be carefully controlled in multi-cell random access networks.
In this paper, we address a challenging and fundamental issue of multi-cell random access for MTC and IoT environments.

\subsection{Related Work} \label{subsec_relatedWork}
Interference management of cellular networks with multiple base stations has been studied extensively \cite{wyner1994shannon, somekh2000shannon}.
While there have been many attempts to develop the optimal scheme with respect to the information-theoretic capacity in multiuser cellular networks, interference alignment (IA) was recently introduced for fundamentally solving the interference problem in multiuser interference channels \cite{cadambe2008interference}.
It was shown that IA can asymptotically achieve the optimal degrees of freedom, which are equal to $\frac{K}{2}$, in the $K$-user interference channel with time-varying fading.
Follow-up studies showed that the interference management based on IA can be applicable to a variety of communication scenarios \cite{gou2010degrees, cadambe2008degrees, yang2017opportunistic} such as interfering multiple access networks \cite{suh2008interference, jung2011opportunistic, jung2012opportunistic}.
Besides the multiple access scenarios in which there is no collision, it is of paramount importance to design an interference management protocol in \emph{random access}.
For multi-cell random access scenarios, there have been several attempts to manage interference by performing IA \cite{gollakota2009interference, lin2011random, jin2015opportunistic} or successive interference cancellation (SIC) \cite{jakovetic2015cooperative, zhou2014bbn}.
In \cite{lin2015simple, lin2016decentralized}, decentralized power allocation  was shown by means of interference mitigation for random access with functionalities of multi-packet reception and SIC at the receiver end. Besides the studies in~\cite{lin2015simple, lin2016decentralized}, there have been a variety of attempts to achieve collision avoidance and interference management in multiuser communication scenarios. Specifically, to reduce the amount of interference, a receiver-oriented channel allocation algorithm for multi-channel device-to-device (D2D) communications was presented in~\cite{HZhao_TVT}; a Sender-Jump Receiver-Wait algorithm for blind rendezvous in cognitive ratio networks (CRNs) was designed in~\cite{JLi_TMC}; a bio-inspired algorithm for efficient spectrum allocation in cognitive IoT networks was proposed in~\cite{JLi_IoTJ}; an access delay model for CRNs was formulated in~\cite{JLi_TCOM} from the cross-layer perspective by employing channel-hopping-based multiuser contention; joint imperfect spectrum sensing
and multi-channel access contention for multi-channel CRNs were shown in~\cite{SZhang_TWC} from the cross-layer perspective; and a collision-avoidance MAC protocol for distributed ad hoc networks was presented in~\cite{HZhao_JNCA}, converging to a collision-free network.

On the other hand, there have been a great deal of studies on exploiting the {\it multiuser diversity} gain in single-cell broadcast networks with a sufficiently large number of users, where opportunistic scheduling \cite{knopp1995information}, opportunistic beamforming \cite{viswanath2002opportunistic}, and random beamforming \cite{sharif2005capacity} were developed. It was also investigated how to exploit the multiuser diversity gain in multi-cell environments. In particular, the work of \cite{sharif2005capacity} on opportunistic transmission can be extended to multi-cell broadcast networks.
As a more challenging task than the downlink case, for multi-cell multiple access, it was shown that if scheduling criteria are properly determined and the number of per-cell users is greater than some level, then the full multiuser diversity gain can be achieved by a distributed user scheduling strategy, thereby leading to the optimal throughput scaling~\cite{shin2012can}.
In addition to the above multiple access scenarios, the multiuser diversity gain can also be exploited in random access networks where there may exist a collision.
The notion of single-cell opportunistic random access (ORA) (dubbed channel-aware slotted ALOHA in the literature) was presented for slotted ALOHA random access networks deploying a single AP \cite{qin2006distributed}.
By assuming that channel state information (CSI) are available at the transmitters, the ORA protocols in \cite{qin2006distributed} were shown to achieve the multiuser diversity gain without any centralized scheduling.
This idea was extended to various scenarios, including slotted ALOHA random access networks with imperfect CSI \cite{wang2009transmission}, a scenario with discontinuous channel measurements \cite{moon2016channel}, carrier sense multiple access networks \cite{miao2012channel}, and multichannel wireless networks \cite{Khanian2016, liu2015stay}.
Nevertheless, the single AP problem was handled in all the protocols~\cite{qin2006distributed, adireddy2005exploiting, wang2009transmission, moon2016channel, miao2012channel, Khanian2016, liu2015stay}; thus, it is not straightforward to apply the idea of ORA to \emph{multi-cell} random access networks where there exists inter-cell interference.

\subsection{Motivation and Contributions}
In this paper, we consider an ultra-dense  time-division duplex (TDD) $K$-cell slotted ALOHA \emph{random access} network, consisting of one AP and $N$ users in \emph{each} cell, which is suited for MTC and IoT networks. We focus primarily on rather simple slotted ALOHA instead of sophisticated MAC protocols since we aim at designing an intuitive random access protocol that leads to the $K$-fold MAC throughput gain in our network. To this end, in the network model with decentralized transmission, we propose an \emph{interference-aware opportunistic random access (IA-ORA)} protocol that achieves the optimal \emph{throughput scaling} (i.e., sum rate scaling) by effectively exploiting the multiuser diversity gain.
Precisely, not only the MAC throughput (or equivalently, the network throughput) of $\frac{K}{e}$ but also the power gain of $\log \log N$ can be achieved in the $K$-cell slotted ALOHA random access network under consideration.
Here, the {\em MAC throghput} of slotted ALOHA is defined as the average number of successfully decoded packets per time slot in the literature, extending the performace metric in \cite{li2016maximum} to multi-cell environments.
First of all, it is worth noting that our IA-ORA protocol is fundamentally different from the two different types of opportunistic transmission protocols in Section \ref{subsec_relatedWork} from the following perspectives:
\begin{itemize}
  \item In contrast to \emph{opportunistic scheduling} in \cite{knopp1995information, viswanath2002opportunistic, sharif2005capacity, shin2012can} applied for cellular multiple access in which base stations select users based on feedforward information, both the intra-cell collision and the inter-cell interference can be mitigated only by users' \emph{opportunistic transmission} in our decentralized IA-ORA protocol designed for random access.
  \item While the conventional ORA protocol in \cite{qin2006distributed} was shown to achieve the multiuser diversity gain (i.e., the power gain) for single-AP slotted ALOHA random access, its extension to multi-cell random access is not straightforward since there exists the inter-cell interference.
Moreover, besides the power gain, it remains open how to offer the \emph{$K$-fold increase} in the MAC throughput via appropriately mitigating the inter-cell interference in the $K$-cell slotted ALOHA random access network.
\end{itemize}

In our ultra-dense multi-cell random access network, users in each cell contend for the same channel at random without centralized coordination from the serving AP.
Consequently, without carefully designing a random access protocol, it is impossible to entirely avoid the inter-cell interference as well as the intra-cell collision, which may result in a failure of packet decoding at the receivers.
In our network model, each user needs to decide whether to transmit or not by itself, without any centralized coordination from the serving AP.
The nature of such random access imposes another difficulty on the protocol design.
We thus aim to respond to these challenges by introducing the IA-ORA protocol.
To this end, we first assume that uplink \emph{channel amplitude information (CAI)} to multiple APs is available at the transmitters owing to the reciprocity between uplink and downlink in TDD mode.
To design our protocol, we utilize this \emph{partial} CAI at the transmitter (CAIT) that can be acquired through a small amount of feedback information sent by APs (to be specified in Section \ref{sec_model}).
In the initialization phase, system parameters such as two thresholds and a physical layer (PHY) data rate are computed \emph{offline} and are broadcast over the network.
Afterwards, each user in a cell first estimates the uplink CAI through the downlink channel in each time slot.
Then, each user determines whether both 1) the channel gain to the belonging AP is higher than one threshold and 2) the total inter-cell interference leakage generated by this user to the other APs is lower than another threshold.
Users transmit their data along with the given PHY data rate opportunistically if the above two conditions hold.
With the help of such opportunistic transmission, when $N$ is sufficiently large in our ultra-dense random access setup, the receivers are able to successfully decode their desired packets sent from multiple users with high probability even in the presence of inter-cell interference, while guaranteeing the optimal \emph{throughput scaling}.
Note that during the communication phase, no control signaling from the APs is required, i.e., all the users \emph{independently} perform opportunistic transmission.
To the best of our knowledge, multi-cell random access in a PHY perspective has not been well investigated before in the literature.

Our main results are four-fold and summarized as follows.
\begin{itemize}
 \item A IA-ORA protocol operating in a decentralized fashion is presented in the ultra-dense $K$-cell slotted ALOHA random access network, where both the desired signal power to the serving AP and the generating interference to other APs are judiciously leveraged.
  \item In the network model, it is shown that the aggregate throughput achieved by our IA-ORA protocol scales as $\frac{K}{e} (1-\epsilon) \log (\textsf{snr} \log N)$ in a high signal-to-noise ratio (SNR) regime, provided that $N$ is larger than $\textsf{snr}^{\frac{K-1}{1-\delta}}$ with respect to \textsf{snr} for an arbitrarily small constant $\epsilon>0$ and a constant $0<\delta<1$.
This implies that the proposed IA-ORA protocol is capable of achieving the full multiuser diversity gain as well as the MAC throughput of $\frac{1}{e}$ in a cell, which is the best we can hope for in slotted ALOHA-type random access \cite{Bertsekas1992}.
  \item Our analysis is validated by numerically evaluating the aggregate throughput through intensive computer simulations.
  We evaluate the throughput in \emph{feasible} $N$ regimes.
  Under practical settings, it is also shown that our IA-ORA protocol outperforms the conventional ORA protocol designed for single-AP random access in terms of aggregate throughput for almost all realistic SNR regimes.
  Additionally, to investigate the robustness of our IA-ORA protocol in the presence of channel uncertainty, we perform simulations under the assumption of imperfect partial CAIT.
  It is examined that the IA-ORA protocol with the imperfect partial CAIT achieves comparable performance on the aggregate throughput to the case with the perfect partial CAIT if the amount of uncertainty is below a certain tolerable level, while still outperforming the conventional ORA protocol.
    \item We extend our ORA protocol to a CSMA/CA random access network by presenting a new CSMA/CA protocol with opportunistic BO such that the BO timer of the user with a higher channel gain is accelerated. To examine a meaningful trade-off between the $K$-fold MAC throughput gain by IA-ORA and the collision avoidance by CSMA/CA, we numerically evaluate the performance under practical settings.
\end{itemize}

Some interference management protocols for multi-cell random access that employ IA and/or the AP cooperation \cite{lin2011random, zhou2014bbn} showed implementation successes based on industrial standards such as IEEE 802.11.
Our IA-ORA protocol also sheds important insights into a simple implementation of multi-cell random access, since neither the dimension expansion nor the AP cooperation is required.

\subsection{Organization and Notations}
Section \ref{sec_model} presents system and channel models.
Section \ref{sec_method} describes the proposed IA-ORA protocol.
Section \ref{sec_analysis} shows how to derive the optimal aggregate throughput scaling and the corresponding user scaling law.
Section \ref{sec_result} provides numerical results to validate our analysis.
Section \ref{sec_CSMA/CA} presents an extension of our ORA protocol to CSMA/CA random access.
Section \ref{sec_conclusion} summarizes the paper with some concluding remarks.

\begin{table}[t!]
  \begin{center}
    \caption{Summary of notations.}
    \label{tab_notation}
    \begin{tabular}{l|l}
      \hline
      \textbf{Notation} & \textbf{Description}\\
      \hline
      $K$ & number of cells\\
      $N$ & number of users in a cell\\
      $h_{j \rightarrow k}^i$ &channel coefficient from the $i$th user in the $j$th cell\\
      &to the $k$th AP\\
      $g_{j \rightarrow k}^i$ &channel gain from the $i$th user in the $j$th cell \\
      &to the $k$th AP\\
      $\Phi_G$ &the first threshold associated with the channel gain\\
      & to the belonging AP\\
      $\Phi_I$ &the second threshold associated with the sum of channel gains\\
      & to the other APs\\
      $R$ &PHY data rate for each user\\
      $p$ &transmission probability\\
      $p_s$ &successful decoding probability\\
      $R_\text{sum}$ &aggregate throughput\\
      \hline
    \end{tabular}
  \end{center}
\end{table}

Throughout the paper, $\mathbb{C}$ denotes the field of complex numbers.
We use the following asymptotic notation: $g(x) = \Omega(f(x))$ means that there exist constants $C$ and $c$ such that $g(x) \leq Cf(x)$ for all $x>c$.
Table I summarizes the notations used in this paper. These notations will be formally defined in the following sections when we introduce our system model and technical details.

\section{System and Channel Models} \label{sec_model}
\begin{figure}
\begin {center}
\epsfig{file=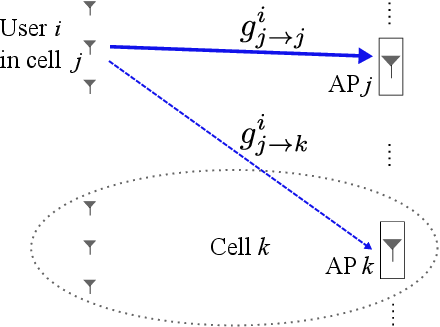, width=0.35\hsize}
\end {center}
\caption{The system model of an ultra-dense $K$-cell random access network with one AP and $N$ users in each cell for large $N$.}
\label{fig_model}
\end{figure}
Let us consider an ultra-dense TDD $K$-cell random access network deploying $K \geq 1$ APs using the same frequency band as shown in Fig.~\ref{fig_model}, where $N$ users are served in each cell and $N$ is sufficiently large. A slotted ALOHA-type protocol is adopted.
No cooperation is assumed among the APs for decoding, i.e., each AP attempts to independently decode the received packets sent from the belonging users.\footnote{Unlike \cite{jakovetic2015cooperative}, we do not assume cooperative decoding among APs, but our developed idea can be extended to to another framework allowing cooperation among APs as future work.}
All the users and APs are equipped with a single antenna. We also make the following assumptions: 1) perfect slot-level synchronization not only between the users and the serving AP but also among the APs;
2) fully-loaded traffic such that each user has a non-empty queue of packets to transmit, similarly as in \cite{li2016maximum}; and 3) transmission of a head-of-line packet with probability $p$ at random, regardless of the number of retransmissions, i.e., each packet is assumed to be the same for all retransmission states.
We adopt a modified signal-to-interference-plus-noise ratio (SINR) model such that each AP is able to decode the received packet if the received SINR exceeds a given decoding threshold, while treating the inter-cell interference (the interfering signals from the inter-cell users) as noise.
If collision occurs due to the concurrent intra-cell transmission when two or more users in the same cell simultaneously transmit, then the corresponding receiver (AP) fails to decode any packet.
That is, in order to simplify system modeling and protocol design, we do not adopt multi-packet reception and multiuser detection studied in \cite{zhang2009does,bae2014achieving}.

Let $\beta_{j \rightarrow k}h_{j \rightarrow k}^i$ denote the channel coefficient from the $i$th user in the $j$th cell to the $k$th AP, which consists of the large-scale path-loss component $0<\beta_{j \rightarrow k} \le 1$ and the small-scale fading component $h_{j \rightarrow k}^i$, following an independent and identically distributed (i.i.d.) complex Gaussian distribution, where $i \in \{1, \cdots, N\}$ and $j, k \in \{1, \cdots, K\}$. For simplicity, we assume that each AP experiences the same degree of path-loss attenuation from the users in the same cell. Especially, when $j=k$, the large-scale term $\beta_{j \rightarrow k}$ is assumed to be 1 since it corresponds to the intra-cell received signal strength, which is much stronger than the inter-cell interference. In the UDN scenario where users in a cell are \emph{densely located}, the large-scale path-loss components from users to APs are assumed to be almost identical. This channel model was also adopted in other studies dealing with throughput scaling analysis in multi-cell UDNs (refer to~\cite{shin2012can,Shin2012network}).

Instead of global CAIT, we assume the partial uplink CAIT that can be acquired via a small amount of feedback information sent by multiple APs; for example, the channel gain from the $i$th user in the $j$th cell to the $k$th AP, denoted by $g_{j \rightarrow k}^i =\beta_{j\rightarrow k}^2 |h_{j \rightarrow k}^i|^2$, is available at the $i$th user.
The partial CAIT $g_{j \rightarrow k}^i$ for $k\in\{1,\cdots,K\}$ can be acquired via the downlink channel due to the reciprocity between uplink and downlink in TDD mode.
Practical CSI acquisition methods were presented in many wireless environments such as multi-cell multi-antenna systems \cite{komulainen2013effective,rogalin2014scalable} and wireless local area networks \cite{IEEE802.11ac, liao2014mu}.
In the seminal study on the ORA with one AP deployment \cite{qin2006distributed}, the partial CAIT and the corresponding distribution information were assumed to be available.
Recently, the CAIT acquisition process was introduced for the multichannel ORA \cite{Khanian2016}.

It is worthwhile to note that since only the \emph{amplitude} information (but not the phase information that needs to be estimated more accurately for satisfactory performance) is required for our protocol, the length of feedback messages can be remarkably shortened via multi-bit quantization, thus resulting in negligible protocol overheads.
Inspired by \cite{qin2006distributed, adireddy2005exploiting, Khanian2016}, an offline channel amplitude acquisition strategy is described under our multi-AP model (i.e., our multi-cell random access model) as follows.
In slotted ALOHA, each AP feeds the short signaling $(0, 1, e)$ back to inform the belonging users of the reception status after each time slot via the downlink channel, where $0$ indicates no packet reception (idle); $1$ indicates reception of only one packet (i.e., successful transmission); and $e$ indicates collision since two or more packets are transmitted simultaneously.
In our multi-cell random access protocol, each AP feeds the short signaling message back in an orthogonal mini-time slot, which needs only a very small amount of coordination among the APs.
That is, coordination among the APs (but not between APs and users) is allowed only for this short feedback message transmission phase so as to avoid collisions.
By exploiting the channel reciprocity between uplink and downlink in TDD systems, each user is capable of estimating the channel amplitudes to multiple APs through the received short signaling messages due to the fact that there inherently exists such a signaling message in the slotted ALOHA transmission protocol \cite{liao2014mu}.

As in the prior studies on the opportunistic transmission \cite{knopp1995information, viswanath2002opportunistic, sharif2005capacity, Shin2012network, shin2012can, qin2006distributed, adireddy2005exploiting, wang2009transmission, moon2016channel, miao2012channel, Khanian2016, liu2015stay}, we assume a quasi-static fading model in which the channel coefficients are constant during one time slot and vary independently in the next time slot.
The received signal $y_k \in \mathbb{C}$ at the $k$th AP is given by
\begin{equation}
\label{eq_rx_signal}
y_k = \underbrace{ \sum_{u_k=1}^{n_k}{ h_{k \rightarrow k}^{\pi(u_k)} x_{k}^{\pi(u_k)}} }_\text{desired signal} + \underbrace{ \sum_{j=1 \atop j\neq k}^{K}{\sum_{u_j=1}^{n_j}{\beta_{j\rightarrow k} h_{j \rightarrow k}^{\pi(u_j)} x_{j}^{\pi(u_j)}} } }_\text{inter-cell interference} + z_k,
\end{equation}
where $x_{k}^{\pi(u_k)}$ is the transmit signal from the $\pi(u_k)$th user in the $k$th cell; the random variable $n_k \sim B(N, p)$ is a binomial integer representing the number of transmitting users in the $k$th cell, and $z_k \in \mathbb{C}$ is the i.i.d. complex additive white Gaussian noise (AWGN) with zero-mean and the variance $N_0$.
For each transmission, there is an average transmit power constraint $\mathbb{E}\left[ \left|x_{k}^{\pi(u_k)}\right|^2 \right] \leq P_\text{TX}$.
The average SNR at each receiver is defined as $\textsf{snr} \triangleq \frac{P_\text{TX}}{N_0}$.

\section{IA-ORA}
\label{sec_method}
In this section, we describe the entire procedure of our IA-ORA protocol, which includes how to select system parameters.
Under the proposed IA-ORA protocol, users opportunistically transmit with a properly selected PHY data rate $R$ if not only the channel gain to the belonging AP exceeds a given threshold $\Phi_G$ but also the sum of channel gains to the other APs is below another given threshold $\Phi_I$.
One of basic ideas behind our protocol is to share the {\em common} data rate $R$ with all users since it may not be possible to set different $R$'s for all users due to the unexpected amount of inter-cell interference at each AP.
Compared to the conventional ORA protocols that aim at enhancing the desired signal power by only using $\Phi_G$, in our ultra-dense multi-cell setup, confining the inter-cell interference leakage to a given sufficiently low level is more crucial to achieve the MAC throughput.
Hence, provided with the channel gains to other-cell APs, we introduce another threshold $\Phi_I$, which is used to mitigate the inter-cell interference leakage by exploiting the opportunism.
Our IA-ORA protocol operates in a decentralized fashion without any additional control signaling from the APs after the initialization phase, and hence fits well into ultra-dense random access networks.
The detailed description of the protocol is elaborated on in the following subsections.

\subsection{Overall Procedure}
\begin{figure}
\begin {center}
\epsfig{file=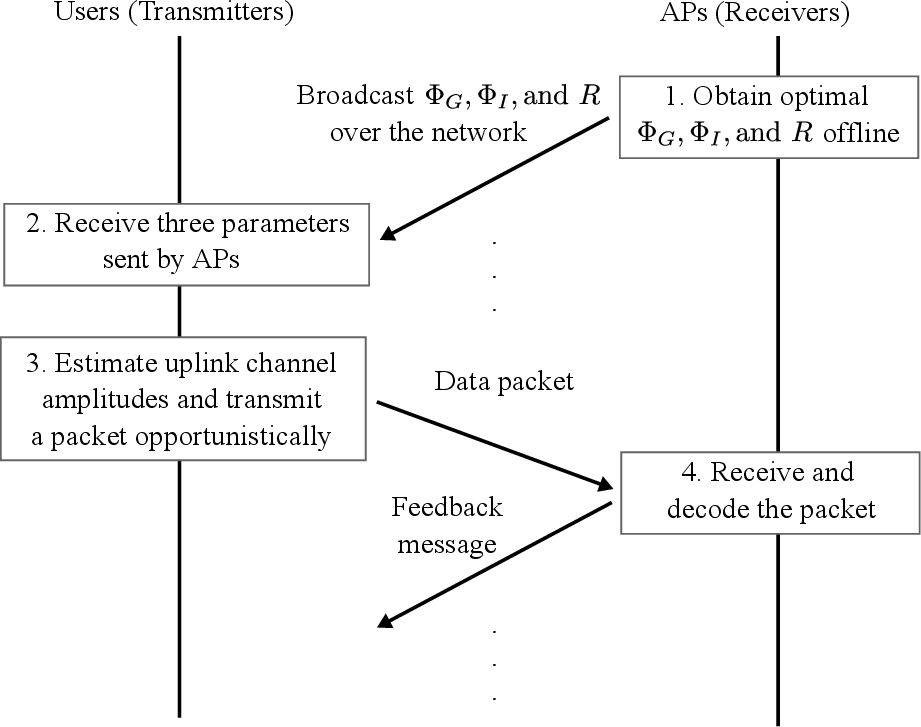, width=0.49\hsize}
\end {center}
\caption{The initialization phase of our IA-ORA protocol.}
\label{fig_flow_diagram}
\end{figure}

In this subsection, we describe the overall procedure of the proposed IA-ORA protocol.
In the initialization phase, the APs share two thresholds $\Phi_G$ and $\Phi_I$ as well as the PHY data rate $R$ with all users for opportunistic transmission in our ultra-dense multi-cell random access network, as illustrated in Fig. \ref{fig_flow_diagram}.
Let us turn to the data communication phase. It is known that the maximum MAC throughput of the conventional slotted ALOHA protocol deploying one AP is achieved at the transmission probability $p = \frac{1}{N}$ for large $N$ (see, e.g., the MAC throughput performance for $N=100$ in Fig. \ref{fig_T}).\footnote{Since only one user is allowed to transmit in a cell, the MAC throughput is expressed as $Np (1-p)^{N-1}$. The probability $p$ maximizing the MAC throughput is given by $p=\frac{1}{N}$.}
Similarly, in our IA-ORA protocol, the transmission probability $p$ averaged out over multiple time slots is also set to $\frac{1}{N}$ to balance between excessive intra-cell collisions or idle time slots, which enables us to discover a relationship between $\Phi_G$ and $\Phi_I$ (to be discussed in Section \ref{subsecOTD}). However, different from the conventional slotted ALOHA, each user makes a decision on whether to transmit depending on its channel conditions.

\begin{figure}
\begin {center}
\epsfig{file=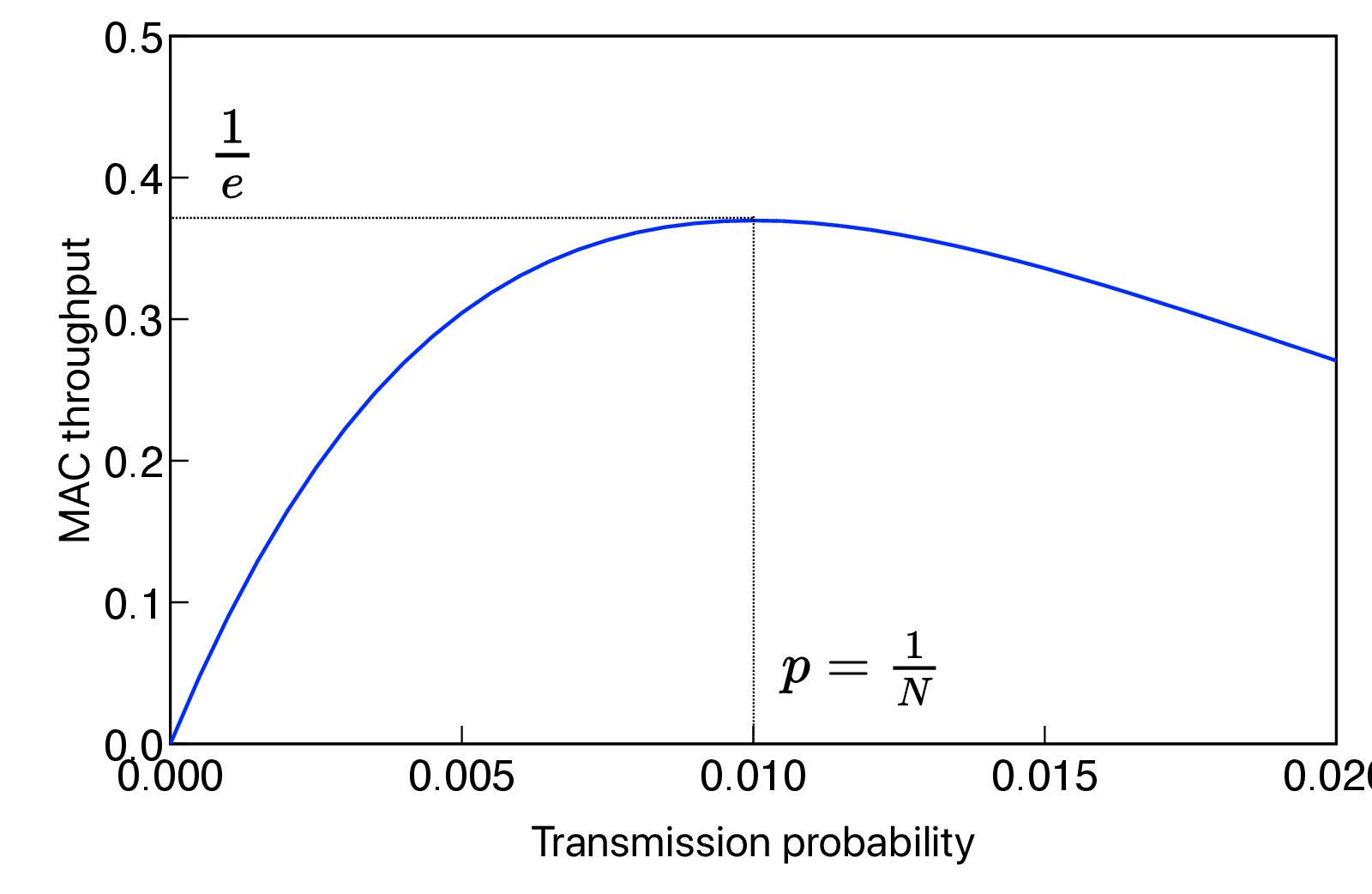, width=0.48\hsize}
\end {center}
\caption{The MAC throughput versus transmission probability $p$, where $N=100$ and the conventional slotted ALOHA with one AP deployment is assumed.}
\label{fig_T}
\end{figure}

In each time slot, each user first estimates the uplink CAIT by using the short signaling feedback messages sent from the APs.
Then, for $i \in \{1, \cdots, N\}$ and $j \in \{1, \cdots, K\}$ (i.e., for all the users), the $i$th user in the $j$th cell compares the channel gains with the two thresholds $\Phi_G$ and $\Phi_I$ to see whether the following two inequalities are fulfilled:
\begin{equation}
\label{eq_thre_G}
g_{j \rightarrow j}^i \geq \Phi_G
\end{equation}
and
\begin{equation}
\label{eq_thre_I}
\sum_{k=1 \atop k \neq j}^{K}{  g_{j \rightarrow k}^i  } \leq \Phi_I,
\end{equation}
where (\ref{eq_thre_G}) indicates a ``good'' channel condition to the serving AP, corresponding to a large desired signal power, and (\ref{eq_thre_I}) indicates a ``weak'' channel condition to the other APs since the inter-cell interference leakage generated by this user is well confined.
From (\ref{eq_thre_G}) and (\ref{eq_thre_I}), it can be seen that only the \emph{amplitude information} of the channels is required.
In each cell, users satisfying both (\ref{eq_thre_G}) and (\ref{eq_thre_I}) start to transmit with the PHY data rate $R$ (to be selected in Section \ref{subsecDRD}), while the other users keep idle in this time slot.
Finally, each AP receives and decodes the desired packet while treating all the interference as noise.
Thanks to the opportunistic transmission, when $N$ is sufficiently large in our ultra-dense random access network, the receivers are capable of successfully decoding their desired packets that were simultaneously sent from multiple users with high probability.

We remark that our IA-ORA protocol operates for general values of $K$.
As a special case, when $K=1$, each user only checks whether the condition (\ref{eq_thre_G}) holds, which corresponds to the conventional ORA protocol.

\subsection{Selection of Two Thresholds} \label{subsecOTD}
In this subsection, we describe how to select the thresholds $\Phi_G$ and $\Phi_I$.
First, according to the two conditions (\ref{eq_thre_G}) and (\ref{eq_thre_I}), the probability that each user succeeds in accessing the channel is given by
\begin{equation} \nonumber
\begin{split}
p &= \Pr\left( g_{j \rightarrow j}^i \geq \Phi_G, \sum_{k = 1 \atop k \neq j}^{K}{  g_{j \rightarrow k}^i  } \leq \Phi_I \right)\\
&= \Pr\left( g_{j \rightarrow j}^i \geq \Phi_G\right) \Pr\left( \sum_{k = 1 \atop k \neq j}^{K}{  g_{j \rightarrow k}^i  } \leq \Phi_I \right),
\end{split}
\end{equation}
where the second equality comes due to the fact that the channel gains to different APs are independent of each other.
Since $p$ is set to $\frac{1}{N}$, we have
\begin{equation} \nonumber
\label{eq_prob_con}
\Pr\left( g_{j \rightarrow j}^i \geq \Phi_G\right) \Pr\left( \sum_{k = 1 \atop k \neq j}^{K}{  g_{j \rightarrow k}^i  } \leq \Phi_I \right) = \frac{1}{N}.
\end{equation}
Then, the relationship between $\Phi_G$ and $\Phi_I$ is given by
\begin{equation} 
\label{eq_GI_relation}
\Phi_G = F_G^{-1} \left(  1 - (F_I(\Phi_I) N )^{-1}  \right),
\end{equation}
which results in $(1-F_G(\Phi_G))F_I(\Phi_I)=\frac{1}{N}$, where $F_G$ and $F_I$ denote the cumulative distribution functions (CDFs) of $g_{j \rightarrow j}^i$ and $\sum_{k = 1 \atop k \neq j}^{K}{  g_{j \rightarrow k}^i  }$, respectively.

As a vital step of our IA-ORA protocol design, the threshold $\Phi_I$ is set to $\textsf{snr}^{-1}$ so that the optimal throughput scaling (i.e., the MAC throughput of $\frac{1}{e}$ in a cell and the power gain) can be achieved with increasing $\textsf{snr}$ for our slotted ALOHA random access (to be proven in Section \ref{sec_analysis}).
Consequently, the two thresholds can be written as
\begin{equation}  \nonumber
\left\{
\begin{aligned}
&\Phi_I = \textsf{snr}^{-1}\\
&\Phi_G = F_G^{-1} \left(  1 - (F_I(\textsf{snr}^{-1}) N )^{-1}  \right).
\end{aligned}
\right.
\end{equation}

\subsection{Selection of PHY Data Rate} \label{subsecDRD}
In this subsection, we describe how to select the PHY data rate $R$ for each user in terms of achieving the optimal throughput scaling law.
We start by computing the resulting successful decoding probability.
Even if an AP receives only one packet from one of the belonging users, the received packet may still be corrupted by not only the noise but also the inter-cell interference.
Thus, the received SINR of the desired packet should exceed a certain decoding threshold, expressed as $2^R - 1$, for successful decoding.
Then, the successful decoding probability $p_s$ is given by

\begin{equation}
\label{eq_succDecPr}
p_s = \Pr\left( \frac{P_\text{TX} g^i_{j \rightarrow j}}{N_0 + \sum\limits_{k=1\atop k \neq j}^{K}{ \sum\limits_{u=0}^{n_k}{P_\text{TX} g_{k \rightarrow j}^{\pi(u)}}} } > 2^R - 1 \right),
\end{equation}
where $n_k \sim B(N, p)$ denotes a binomial random variable representing the number of simultaneously transmitting users in the $k$th cell and $\pi(u)$ denotes the index of transmitting users in each cell.
Using the successful decoding probability in (\ref{eq_succDecPr}), the throughput at the $j$th AP for $j \in \{1, \cdots, K\}$, denoted by $R_\text{sum}^{(j)}$, is given by
\begin{equation}  \nonumber
\label{eq_}
R_\text{sum}^{(j)} = \underbrace{ Np (1-p)^{N-1} }_\text{MAC throughput} \cdot R \cdot p_s,
\end{equation}
where $Np (1-p)^{N-1}$ is the MAC throughput and $R$ is the target PHY data rate.
From the fact that $p = \frac{1}{N}$, the aggregate throughput of the $K$-cell random access network is expressed as
\begin{equation} \label{eq_sum_rate}
\begin{split}
R_\text{sum} &= \sum_{j=1}^{K}{ R_\text{sum}^{(j)} }\\
&=K \left(1-\frac{1}{N} \right)^{N-1} \cdot R \cdot p_s \\
&\geq K \left(1-\frac{1}{N} \right)^{N-1} \cdot R \cdot \Pr\left( \frac{\Phi_G}{\textsf{snr}^{-1} + \tilde{n}\Phi_I} > 2^R - 1 \right),
\end{split}
\end{equation}
where $\tilde{n} \sim B((K-1)N, p)$ is a binomial random variable, representing the total number of interfering signals from the other cells, and thus is given by $\tilde{n} = \sum_{k = 1 \atop k \neq j}^{K}{ n_k }$.
Here, the inequality in (\ref{eq_sum_rate}) holds due to (\ref{eq_thre_G}) and (\ref{eq_thre_I}).

Now, let us turn to computing a lower bound on the modified successful decoding probability, denoted by $\tilde{p}_s$, shown below:
\begin{equation}
\label{eq_qs}
\begin{split}
\tilde{p}_s &= \Pr\left( \frac{\Phi_G}{\textsf{snr}^{-1} + \tilde{n}\Phi_I} > 2^R - 1 \right)\\
&= \sum_{i=0}^{(K-1)N}{\Pr\left( \frac{\Phi_G}{\textsf{snr}^{-1} + i \Phi_I} > 2^R -1 \right) \Pr(\tilde{n} = i)}.
\end{split}
\end{equation}
Let us introduce an integer $\nu \in \{0, 1, \cdots, (K-1)N \}$.
If $R$ is set to a value such that
\begin{equation} \nonumber
\label{ }
\frac{\Phi_G}{\textsf{snr}^{-1} + (\nu+1)\Phi_I} < 2^R-1 \leq \frac{\Phi_G}{\textsf{snr}^{-1} + \nu \Phi_I},
\end{equation}
then the probability $\Pr\left( \frac{\Phi_G}{\textsf{snr}^{-1} + i \Phi_I} \geq 2^R -1 \right)$ is given by $1$ and $0$ for $i \in \{0, 1, \cdots, \nu\}$ and $i \in \{\nu+1, \cdots, (K-1)N\}$, respectively.
Based on this observation, it is possible to partition the entire feasible range of $2^R-1$ (i.e., the decoding threshold) into the following $(K-1)N + 1$ sub-ranges:
\begin{equation}  \nonumber
\begin{split}
&\left\{ \left(0, \frac{\Phi_G}{\textsf{snr}^{-1} + (K-1)N \Phi_I}\right], \cdots, \right .\\
&\left(\frac{\Phi_G}{\textsf{snr}^{-1} + (\nu+1) \Phi_I}, \frac{\Phi_G}{\textsf{snr}^{-1} + \nu \Phi_I}\right], \cdots, \left . \left(\frac{\Phi_G}{\textsf{snr}^{-1}}, \infty \right) \right\}.
\end{split}
\end{equation}
In particular, for $R \in (\frac{\Phi_G}{\textsf{snr}^{-1}}, \infty)$, we have $\tilde{p}_s=0$, which is thus neglected in our work.

Since the term $\Pr\left( \frac{\Phi_G}{\textsf{snr}^{-1} + i \Phi_I} \geq 2^R -1 \right)$ in (\ref{eq_qs}) is an indicator function of $R$, we set $R$ to the maximum value under the condition such that $2^R-1$ lies in each sub-range $\left(\frac{\Phi_G}{\textsf{snr}^{-1} + (\nu+1) \Phi_I}, \frac{\Phi_G}{\textsf{snr}^{-1} + \nu \Phi_I} \right]$, which is given by
\begin{equation}
\label{eq_discrete_R}
R = \log_2\left( 1 + \frac{\Phi_G}{\textsf{snr}^{-1} + \nu \Phi_I}  \right).
\end{equation}
Here, the parameter $\nu \in \{0, 1, \cdots, (K-1)N \}$ in (\ref{eq_discrete_R}) can be interpreted as a tolerable number of interfering signals from other-cell users.
That is, a relatively high successful decoding probability can be guaranteed even when the desired signal at the receiver is interfered by $\nu$ signals caused by the other-cell users.
Note that we can enhance $\tilde{p}_s$ for each transmitted packet at the loss of $R$ (corresponding to a larger $\nu$).
In other words, we can trade a lower PHY data rate for a higher successful decoding probability (refer to (\ref{eq_qs})).
In Section \ref{subsec_preliminaries}, it will be discussed how to set the value of $\nu$ to determine the PHY data rate $R$ in (\ref{eq_discrete_R}) (refer to Remark 2).

\emph{Remark 1:} It is worth noting that the optimal PHY data rate can be found by solving the aggregate throughput maximization problem as follows: $\hat{R}=\argmax_R R_\text{sum}(R)$.
However, as an alternative approach, it is sufficient to set $R$ to (\ref{eq_discrete_R}) in the sense of guaranteeing the optimal \emph{throughput scaling} because i) selecting a proper finite $\nu$ leads to a successful decoding probability $p_s$ approaching $1$, which plays a crucial role in showing the optimality of our protocol in terms of scaling laws, and ii) the expression in (\ref{eq_discrete_R}) is analytically tractable and thus enables us to derive a closed-form expression of the aggregate throughput scaling law.
Detailed analytical discussions are addressed in the next section.

\section{Analysis of Aggregate Throughput Scaling Law} \label{sec_analysis}
In this section, we show that the proposed IA-ORA protocol asymptotically achieves the optimal throughput scaling, i.e., $\frac{K}{e} (1-\epsilon) \log (\textsf{snr}\log N)$ for an arbitrarily small constant $\epsilon > 0$, in the ultra-dense $K$-cell slotted ALOHA random access network if $N$ is greater than a certain level with respect to \textsf{snr}.
To be specific, we first provide some preliminaries including a simplification of $\tilde{p}_s$ and a tractable lower bound on $F_I(x)$.
We then present our main result by analyzing the aggregate throughput scaling under a certain user scaling condition.
Additionally, we present an upper bound on the aggregate throughput scaling and compare the proposed IA-ORA protocol with the existing opportunistic scheduling approach for multi-cell multiple access networks.

\subsection{Preliminaries} \label{subsec_preliminaries}
In this subsection, we introduce two important lemmas, which play an important role in showing our main result.
To establish the first lemma, we revisit the lower bound on the aggregate throughput in (\ref{eq_sum_rate}).
To simplify this lower bound, we derive an explicit expression of $\Phi_G$ as a function of $\Phi_I$ based on the relationship between $\Phi_G$ and $\Phi_I$ in (\ref{eq_GI_relation}).
Here, we make it more explicitly by using the fact that the channel gain $g_{j \rightarrow j}^i$ follows the exponential distribution whose CDF is given by $F_G(x) = 1 - e^{-x}$, resulting in $F_G^{-1}(x)=\ln(\frac{1}{1-x})$. Thus, the relationship between $\Phi_G$ and $\Phi_I$ can be found as follows:
\begin{equation}
\label{eq_phig_phii}
\begin{split}
\Phi_G &= F_G^{-1} \left(  1 - (F_I(\Phi_I) N )^{-1}  \right)\\
&= \ln \left(  \frac{1}{1 - \left( 1 - (F_I(\Phi_I) N )^{-1} \right)}  \right)\\
&= \ln \left( F_I(\Phi_I) N  \right).
\end{split}
\end{equation}
By using (\ref{eq_phig_phii}) in (\ref{eq_sum_rate}), we have
\begin{equation}
\label{eq_bound2}
\begin{split}
R_\text{sum} \geq &K \underbrace{ \left(1-\frac{1}{N} \right)^{N-1} }_\text{MAC throughput} \cdot R\cdot \underbrace{ \Pr\left( \frac{ \ln \left( F_I(\Phi_I) N  \right) }{\textsf{snr}^{-1} + \tilde{n}\Phi_I} > 2^R - 1 \right) }_{\tilde{p}_s}.
\end{split}
\end{equation}

For analytical convenience, the following lemma is presented due to the analytical intractabilty of the resulting form of $\tilde{p}_s$.

\begin{lemma} \label{lemma_discrete_R_Q}
When the PHY data rate $R$ is chosen from the discrete set in (\ref{eq_discrete_R}), it follows that
\begin{equation}
\label{eq_Q}
\tilde{p}_s = I_{1-\frac{1}{N}} \left( (K-1)N-\nu, \nu+1 \right),
\end{equation}
where $I_{x}(y, z)$ is the regularized incomplete beta function.
\end{lemma}
\begin{proof}[Proof: \nopunct] We refer to Appendix A.
\end{proof}

\begin{figure}
\begin {center}
\epsfig{file=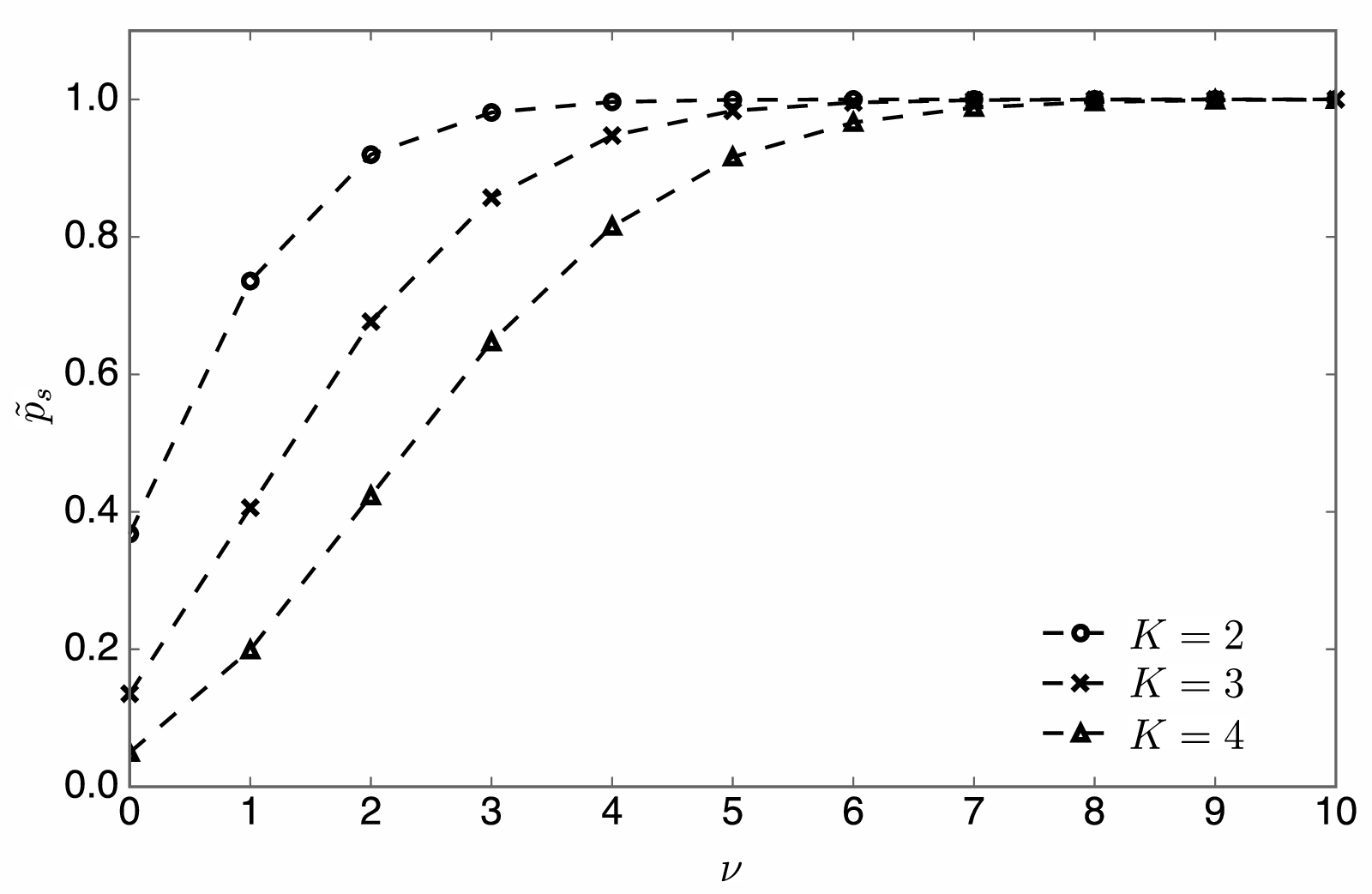, width=0.48\hsize}
\end {center}
\caption{The function $\tilde{p}_s=I_{1-\frac{1}{N}} \left( (K-1)N-\nu, \nu+1 \right)$ versus $\nu$ for $K = \{ 2,3,4 \}$.}
\label{fig_Q}
\end{figure}

From this lemma, it is possible to transform the expression of $\tilde{p}_s$ in (\ref{eq_bound2}) equivalently to the well-known regularized incomplete beta function.
Now, the following interesting observation is made in Remark 2, which enables us to  state how to decide the value of $\nu$.

\emph{Remark 2:} In Fig. \ref{fig_Q}, the function $I_{1-\frac{1}{N}} \left( (K-1)N-\nu, \nu+1 \right)$ is illustrated according to $\nu$ for various $K = \{ 2,3,4 \}$.
This function tends to increase rapidly with $\nu$ for small $\nu$ regimes and then gradually approach one.
That is, it is monotonically increasing with $\nu$; for example, for $K=2$, the pairs of $(\tilde{p}_s, \nu)$ are given by $(0.9, 3), (0.99, 5)$, and $(0.999, 7)$.
This observation can be generalized to any values of $K$.
Thus, we can appropriately select a \emph{finite} value of $\nu$ that makes $\tilde{p}_s$ approach almost one by taking the inverse of the regularized incomplete beta function $I_{1-\frac{1}{N}}(\cdot,\cdot)$.

Using (\ref{eq_discrete_R}), (\ref{eq_bound2}), and (\ref{eq_Q}), the aggregate throughput is now lower-bounded by
\begin{equation}
\label{eq_bound_FI}
\begin{split}
R_\text{sum} \geq &K \left(1-\frac{1}{N}\right)^{N-1} \cdot \underbrace{ \log_2\left( 1 + \frac{ \ln \left( F_I(\Phi_I) N  \right) }{\textsf{snr}^{-1} + \nu \Phi_I}  \right)}_\text{PHY data rate ($R$)}\\
&\cdot I_{1-\frac{1}{N}} \left( (K-1)N-\nu, \nu+1 \right).
\end{split}
\end{equation}

However, motivated by the fact that it is still not easy to find a closed form expression of (\ref{eq_bound_FI}) due to the complicated form of $F_I(\Phi_I)$, we define the CDF of $\sum_{k = 1 \atop k \neq j}^{K}|h_{j \rightarrow k}^i|^2$ as $F_{\tilde{I}}$. Then, since $F_{\tilde{I}}(\Phi_I)\le F_{I}(\Phi_I)$ due to $0<\beta_{j \rightarrow k} \le 1$, the aggregate throughput in (\ref{eq_bound_FI}) is further lower-bounded by
\begin{equation}
\label{eq_bound_FI2}
\begin{split}
R_\text{sum} \geq &K \left(1-\frac{1}{N}\right)^{N-1} \cdot \log_2\left( 1 + \frac{ \ln \left( F_{\tilde{I}}(\Phi_I) N  \right) }{\textsf{snr}^{-1} + \nu \Phi_I}  \right)\\
&\cdot I_{1-\frac{1}{N}} \left( (K-1)N-\nu, \nu+1 \right).
\end{split}
\end{equation}
Now, we introduce the following lemma, which computes a lower bound on $F_{\tilde{I}}(\Phi_I)$ from the fact that the variable $\sum_{k = 1 \atop k \neq j}^{K}|h_{j \rightarrow k}^i|^2$ follows the chi-square distribution with $2(K-1)$ degrees of freedom.
\begin{lemma} \label{lemma_FI_bound}
For any $0 \leq x < 2$, the CDF of $\sum\limits_{k=1 \atop k \neq j}^{K}{ |h_{j \rightarrow k}^i|^2 }$, $F_{\tilde{I}}(x)$, is lower-bounded by
\begin{equation} \nonumber
F_{\tilde{I}}(x) \geq c_1 x^{(K-1)},
\end{equation}
where $c_1 = \frac{e^{-1} 2^{-(K-1)}}{(K-1) \Gamma(K-1)}$ is some constant independent of $N$ and SNR.
Here, $\Gamma(x) = (x-1)!$ is the Gamma function.
\end{lemma}
\begin{proof}[Proof: \nopunct]
We refer to \cite[Lemma 2]{shin2012can} for the proof.
\end{proof}
This lemma enables us to find a more tractable lower bound on the aggregate throughput, which will be analyzed in the next subsection.

\subsection{Aggregate Throughput Scaling and User Scaling Laws} \label{subsec_theorem}
As our main result, we establish the following theorem that characterizes the aggregate throughput scaling achieved by our IA-ORA protocol.

\begin{theorem} \label{theorem_scaling}
Consider the IA-ORA protocol in the ultra-dense $K$-cell slotted ALOHA random access network.
Suppose that $\Phi_I = \textsf{snr}^{-1}$, and $\Phi_G$ and $R$ are set to (\ref{eq_phig_phii}) and (\ref{eq_discrete_R}), respectively.
Then, an aggregate throughput scaling of
\begin{equation} 
\label{eq_scaling}
\frac{K}{e} (1-\epsilon) \log (\textsf{snr}\log N)
\end{equation}
is achieved by the IA-ORA protocol with high probability in the high SNR regime provided that
\begin{equation} 
\label{eq_condition}
N = \Omega \left( \textsf{snr}^{\frac{K-1}{1-\delta}} \right),
\end{equation}
where $\epsilon> 0$ is an arbitrarily small constant and $0<\delta<1$ is some constant.
\end{theorem}

\begin{proof}[Proof: \nopunct] We refer to Appendix B.
\end{proof}

Theorem \ref{theorem_scaling} indicates that our IA-ORA protocol can achieve not only the near $\frac{K}{e}$ MAC throughput (corresponding to the pre-log term in (\ref{eq_scaling})) but also the power gain of $\log \log N$ in multi-cell random access.
In particular, the power gain comes from the threshold setting of $\Phi_G$ in (\ref{eq_thre_G}), which enables us to guarantee a high channel gain to the belonging AP.
To obtain such gains, the inter-cell interference leakage caused by a certain user needs to be well confined by setting $\Phi_I = \textsf{snr}^{-1}$.
This implies that the rigid user scaling condition in (\ref{eq_condition}) is necessary to confine the total interference to a fixed level.
Based on the analytical result in Theorem \ref{theorem_scaling}, the following interesting observations are provided with respect to the parameters $K$, $\delta$, and $\epsilon$.

\emph{Remark 3:} According to (\ref{eq_scaling}) and (\ref{eq_condition}), it is found that the aggregate throughput scaling increases linearly with $K$ at the cost of more strict user scaling law (note that $N$ needs to increase exponentially with $K$ for given $\textsf{snr}$).
In ultra-dense random access networks with many users, our IA-ORA protocol can provide higher MAC throughput by deploying more APs.

\emph{Remark 4:} We now turn to the effect of $\delta$ on the aggregate throughput performance.
Although increasing $\delta$ leads to an increment of the aggregate throughput (refer to (\ref{eq_finalBound}) in Appendix B), it does not fundamentally change the aggregate throughput \emph{scaling law}.
However, increasing $\delta$ yields a higher number of users to guarantee the performance.
Thus, it is sufficient to properly select a value of $0<\delta<1$ according to the given network condition.

\emph{Remark 5:} The pre-log term $1-\epsilon$ represents the successful decoding probability, which can be regarded as a penalty of random access without any coordination.
This penalty cannot be totally resolved because there is always a non-zero probability of unsuccessful decoding caused by the excessive inter-cell interference generated by random transmission (refer to (\ref{eq_Q_star}) in Appendix B).
This is a distinctive phenomenon of multi-cell \emph{random access}, compared to the multi-cell multiple access scenario in which each base station performs user scheduling.

To increase the successful decoding probability, we can decrease $\epsilon$ by selecting a larger $\nu^*$ according to (\ref{eq_Q}).
However, this results in a lower PHY data rate $R$ in (\ref{eq_discrete_R}) that affects the resulting aggregate throughput.
Thus, we need to carefully balance this trade-off between $\epsilon$ and $\nu^*$ when selecting system parameters in practice.
Nevertheless, it is sufficient to assume a finite $\nu^*$ leading to an arbitrarily small $\epsilon>0$ to analyze our throughput scaling result.

\subsection{Discussions} \label{analysis_discussion}
In this subsection, we first present an upper bound on the aggregate throughput scaling that matches our analytically achievable result.
Then, we compare the proposed IA-ORA protocol with the opportunistic scheduling protocol for the multi-cell multiple access \cite{shin2012can}.

\emph{An upper bound on the aggregate throughput scaling:} A genie-aided removal of all the inter-cell interference leads to an aggregate throughput scaling of $\frac{K}{e} \log (\textsf{snr}\log N)$, which is explained as follows. Consider an ideal scenario for the multi-cell random access network where concurrent interfering signals sent from other cells are completely canceled out.
This can be thought of as a system consisting of $K$ \emph{interference-free} parallel random access networks.
Since the throughput of $\frac{1}{e} \log (\textsf{snr}\log N)$ is achieved in each cell of such networks \cite{qin2006distributed}, the aggregate throughput scales as $\frac{K}{e} \log (\textsf{snr}\log N)$.
This corresponds to an upper bound on the aggregate throughput scaling for the $K$-cell random access network, which matches the achievable aggregate throughput scaling in (\ref{eq_scaling}) to within a factor of $\epsilon>0$.

\emph{Comparison with the opportunistic scheduling protocol for the multi-cell multiple access:} For the multi-cell multiple access in which collisions can be avoided by adopting user scheduling at each base station \cite{shin2012can}, the opportunistic scheduling protocol achieves the aggregate throughput scaling of $K \log (\textsf{snr}\log N)$. By comparing our IA-ORA protocol with this scheduling protocol, some interesting consistency can be found as follows:
1) user scaling laws required by the two protocols are identical and 2) under this user scaling condition, the two protocols achieve the same throughput scaling to within an arbitrary small gap $\epsilon>0$ after removing the intra-cell contention loss factor (note that the MAC throughput in a cell is given by $1$ and $\frac{1}{e}$ in multiple access and slotted ALOHA random access, respectively).
Hence, essential similarities are revealed by applying opportunism to both multiple access and random access networks.

In contrast to the observed similarities, there are fundamental differences between the two protocols.
Under the opportunistic scheduling protocol in \cite{shin2012can}, users in each cell who fulfill the given channel conditions send transmission requests to the base station.
Then, each base station randomly selects one user and informs the user of the PHY data rate for transmission via the downlink channel.
On the other hand, under our IA-ORA protocol, without controlling from the APs, each user determines whether to transmit by comparing its own channel gains with the two thresholds.
Besides, the PHY data rate does not need to be sent from the APs since it is already broadcast in the initialization phase and keeps the same afterwards.
These features make our IA-ORA protocol suitable for uncoordinated random access networks.

\section{Numerical Evaluation} \label{sec_result}
In this section, we validate the proposed IA-ORA protocol by intensively performing numerical evaluation through Monte-Carlo simulations, where the channels in (\ref{eq_rx_signal}) are randomly generated $1 \times 10^5$ times for each system parameter.
First, our analysis is validated by comparing the aggregate throughput through numerical evaluation with the corresponding theoretical one.
Performance of the proposed IA-ORA protocol is also compared with that of the single-cell ORA protocol \cite{qin2006distributed} and the slotted ALOHA protocol \cite{Bertsekas1992} in realistic environments.
Additionally, performance of both the IA-ORA and ORA protocols is evaluated under the assumption of imperfect partial CAIT.

\subsection{Validation of Analytical Results}
In this subsection, we evaluate the aggregate throughput~[bps/Hz] of our IA-ORA protocol in Sections \ref{sec_method} and \ref{sec_analysis} to validate our analytical result in an ultra-dense multi-cell random access setup.
By setting $\epsilon = 0.01$ and $\delta = 0.1$, the parameter $\nu$ can be found according to the relationship between two parameters $\epsilon$ and $\nu$ in (\ref{eq_Q_star}) of Appendix B.
We assume that $\Phi_I = \textsf{snr}^{-1}$ and $\Phi_G = \ln(F_I(\Phi_I)N)$ (see (\ref{eq_phig_phii})).
Based on these parameters, the PHY data rate $R$ can be computed from (\ref{eq_discrete_R}).

In Fig. \ref{fig_Rsum_SNR_ana}, the aggregate throughput achieved by the IA-ORA protocol versus $\textsf{snr}$ in dB scale is illustrated for $K = \{2,3\}$.
The parameter $N$ is set to different scalable values depending on $\textsf{snr}$, i.e., $N = \textsf{snr}^{\frac{K-1}{1-\delta}}$ in (\ref{eq_condition}).
In the figure, the dotted lines are also plotted from the analytical result in Theorem \ref{theorem_scaling} with a proper bias to check the slopes of $\frac{K}{e} (1-\epsilon) \log (\textsf{snr}\log N)$ for $K = \{2,3\}$.
It is seen that the slopes of the numerically found curves coincide with the theoretical ones in the high SNR regime. These empirical results are sufficient to justify our analysis (i.e., the throughput scaling under the given user scaling condition) in Section \ref{sec_analysis}.

\begin{figure}
\begin {center}
\epsfig{file=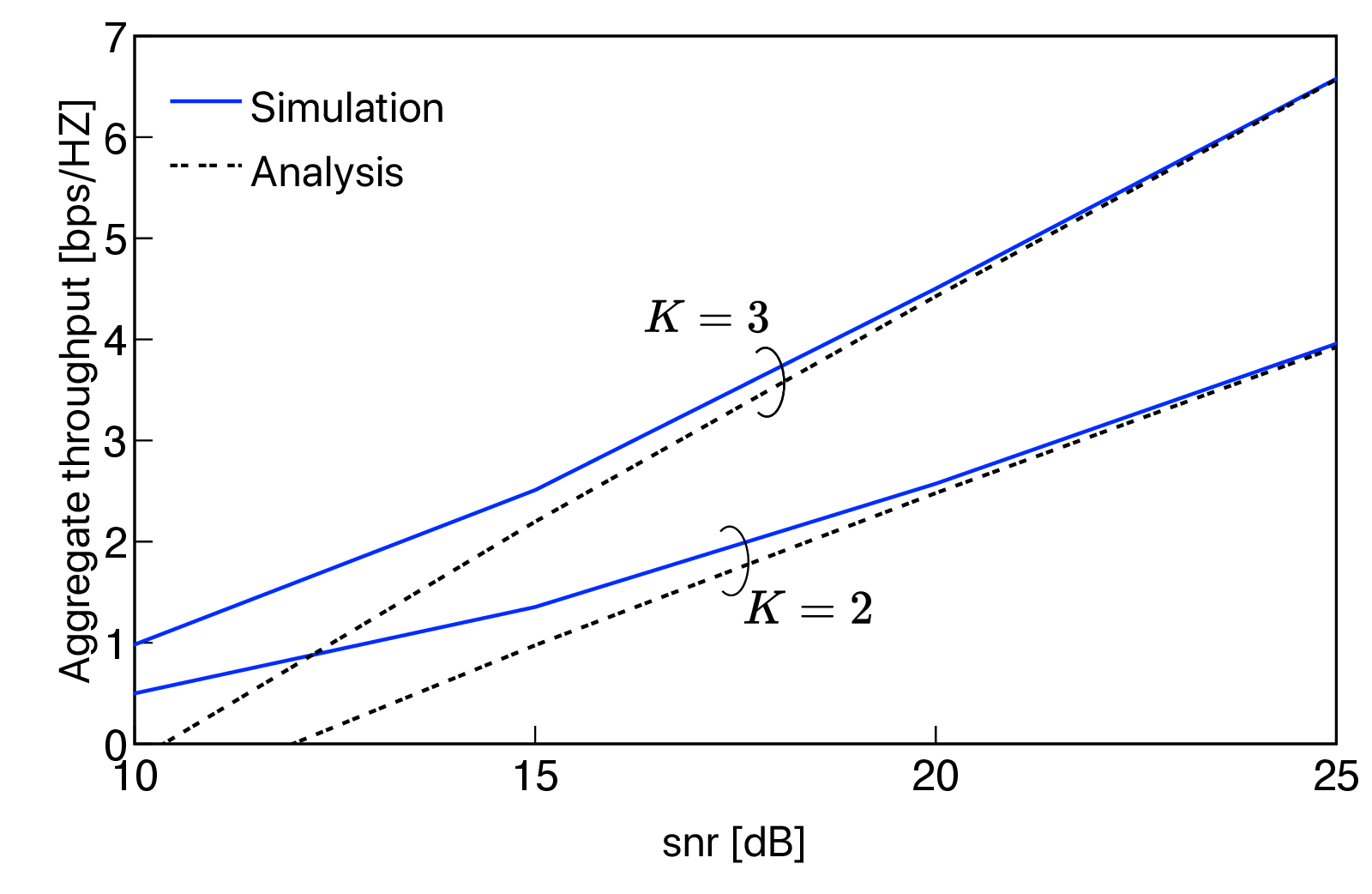, width=0.48\hsize}
\end {center}
\caption{The aggregate throughput versus SNR, where the proposed IA-ORA protocol is employed as $N$ scales according to the user scaling condition.}
\label{fig_Rsum_SNR_ana}
\end{figure}

Now, in Fig. \ref{fig_Rsum_varyN}, the aggregate throughput versus $N$ is illustrated, where $K=2$ and $\textsf{snr} = \{ 10, 15, 20, 25 \}$dB.\footnote{Even if it seems unrealistic to have a large number of users even in ultra-dense random access models, the wide range of parameter $N$ is adopted to precisely see some trends of curves varying with $N$.}
We adopt the aforementioned parameter setting except for $\delta$ that is used to specify the user scaling condition in (\ref{eq_condition}).
One can see that all the curves tend to increase at most \emph{logarithmically} with $N$ owing to the multiuser diversity gain, which is consistent with the throughput scaling in (\ref{eq_scaling}).
It is also observed that for large $N$, increasing the $\textsf{snr}$ leads to superior performance on the aggregate throughput due to the power gain.
As illustrated in the figure, it is worthwhile to note that this superior throughput performance can be achieved when $N$ is sufficiently large (or the user scaling law in (\ref{eq_condition}) is fulfilled).
If $N$ is very small, then increasing the $\textsf{snr}$ does not leads to better performance since the inter-cell interference leakage is not well confined.

\begin{figure}
\begin {center}
\epsfig{file=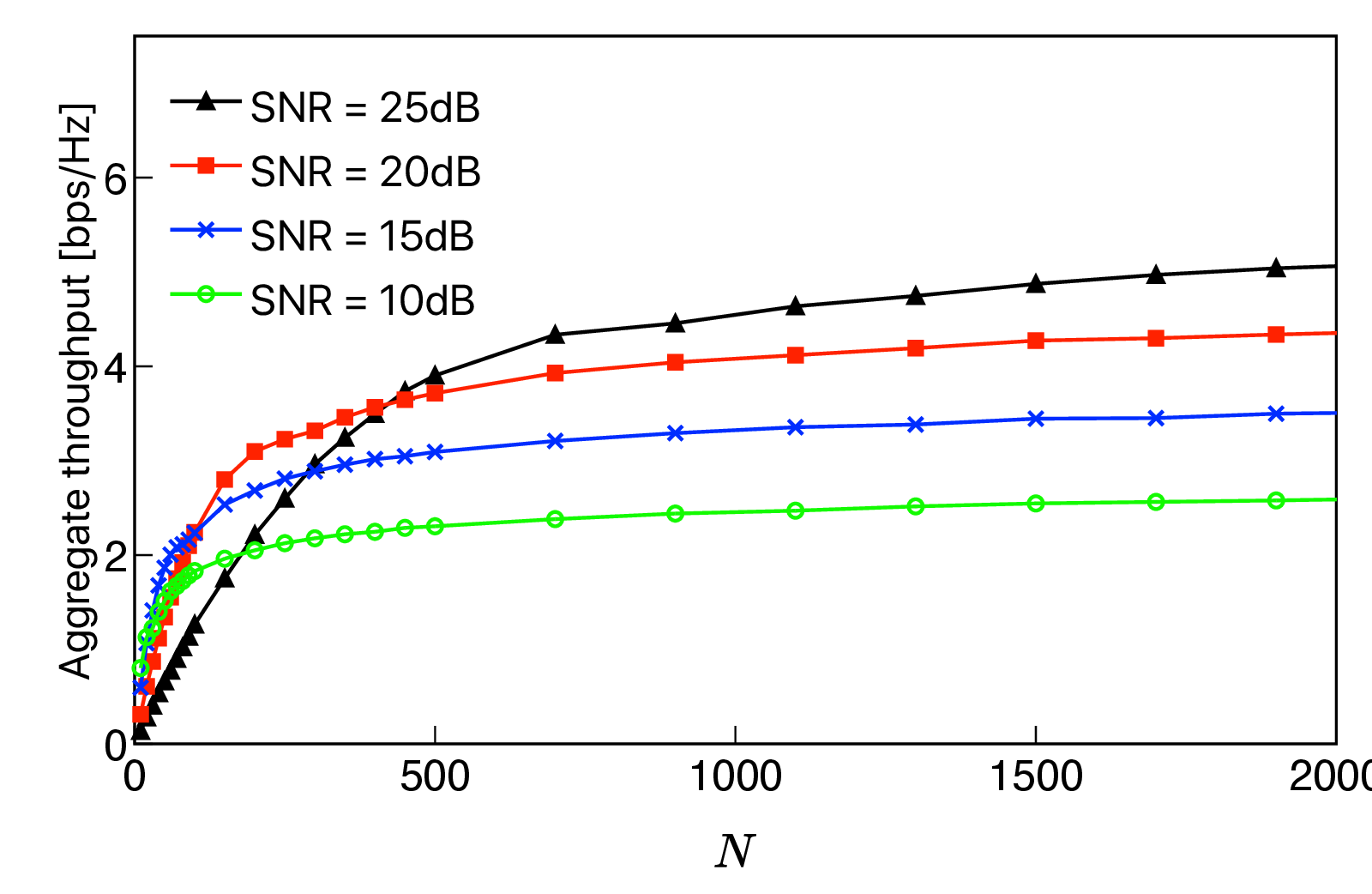, width=0.48\hsize}
\end {center}
\caption{The aggregate throughput versus $N$, where the proposed IA-ORA protocol is employed for $K=2$.}
\label{fig_Rsum_varyN}
\end{figure}

\subsection{Performance Evaluation in Practical Settings}
In this subsection, to demonstrate the effectiveness of our IA-ORA protocol in practice, performance on the aggregate throughput is evaluated in \emph{feasible} $N$ regimes.
More specifically, we evaluate the performance of our protocol when $N$ scales much slower than the condition in (\ref{eq_condition}).
Under this practical setting, instead of using the original IA-ORA protocol, the IA-ORA protocol in Sections \ref{sec_method} and \ref{sec_analysis} is slightly modified in such a way that the optimal parameters $\Phi_G^\ast$ and $R^\ast$ are numerically found in terms of maximizing the resulting aggregate throughput $R_\text{sum}(\Phi_G, R)$.\footnote{Note that the IA-ORA protocol needs to be slighted modified when it is evaluated in practical settings since it is inherently designed for asymptotically achieving the optimal throughput scaling.} That is, we aim to find
\begin{equation} \label{eq_prac_opt}
(\Phi_G^\ast, R^\ast) = \underset{ \Phi_G, R }{\operatorname{argmax}} \quad { R_\text{sum}(\Phi_G, R) }.
\end{equation}
The corresponding $\Phi_I^\ast$ can be found by using the relationship in (\ref{eq_phig_phii}).
The optimal $(\Phi_G^\ast, R^\ast)$ can be found via exhaustive search for given parameter configuration including $K, N$, and $\textsf{snr}$.
The optimal values of $(\Phi_G^\ast, R^\ast)$ are summarized in Table \ref{table_para}.
From the table, some insightful observations are made as follows.
For given $\textsf{snr}$ and $N$, the optimal PHY data rate $R^\ast$ tends to decrease with $K$.
This is because more inter-cell interference is generated for larger $K$, thus leading to a lower SINR at the receivers and the resulting lower $R^\ast$.
Moreover, for given $\textsf{snr}$ and $K$, $R^\ast$ tends to increase with $N$.
This comes from the fact that for larger $N$, the inter-cell interference can be better mitigated due to the multiuser diversity gain, thereby leading to a higher SINR at the receivers.
This enables us to adopt a higher $R^\ast$.
  
  \begin{table}[!t]
  \caption{The optimal values of $(\Phi_G^\ast, R^\ast)$ according to various settings of $K$, $N$, and $\textsf{snr}$.}
  \label{table_para}
  \centering 
  \begin{tabular}{c | c | c | c}
  \hline
  $K \backslash N$ & 50 & 100 & 200\\
  \hline
  \hline
   \multicolumn{4}{ c }{$\textsf{snr}$ = 10dB}\\
   \hline
   2 & (1.40, 2.32) & (1.70, 3.64) & (2.20, 3.91)\\
   3 & (1.70, 2.14) & (1.90, 2.45) & (2.30, 2.72)\\
   4 & (1.90, 1.54) & (2.10, 1.60) & (2.20, 2.02)\\
   \hline
   \multicolumn{4}{ c }{$\textsf{snr}$ = 20dB}\\
   \hline
   2 & (0.90, 4.40) & (0.80, 4.88) & (1.10, 5.61)\\
   3 & (1.70, 2.33) & (1.60, 2.77) & (1.30, 3.18)\\
   4 & (1.70, 1.70) & (1.90, 1.89) & (2.20, 2.09)\\
  \hline
  \end{tabular}
\end{table}

\begin{figure}
\begin {center}
\epsfig{file=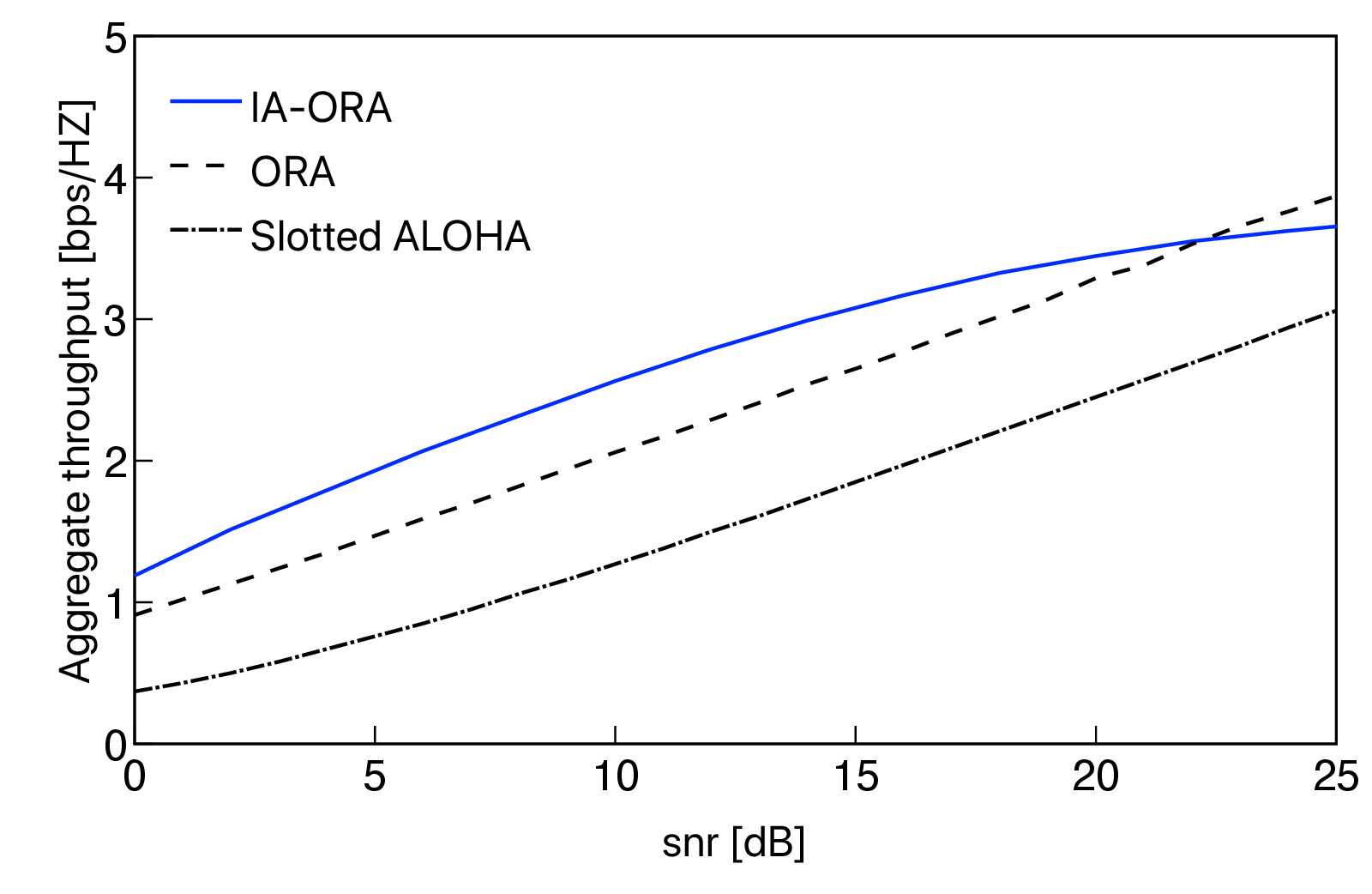, width=0.48\hsize}
\end {center}
\caption{The aggregate throughput versus $\textsf{snr}$, where the proposed IA-ORA protocol ($K=2$) as well as the conventional ORA and slotted ALOHA protocols are employed for $N=100$.}
\label{fig_Rsum_N100_K2_sim}
\end{figure}

In Fig. \ref{fig_Rsum_N100_K2_sim}, the aggregate throughput of our IA-ORA protocol versus $\textsf{snr}$ in dB scale is illustrated for $K=2$ and $N=100$.
As benchmark schemes, performance of the ORA and slotted ALOHA protocols is also illustrated in the figure. In the conventional ORA protocol, each user performs opportunistic transmission with the PHY data rate of $\log_2 \left(1 + \frac{\Phi_G P_\text{TX}}{N_0}\right)$ if its uplink channel gain exceeds $\Phi_G$ \cite{qin2006distributed}. Note that the ORA protocol can be regarded as a special case of our IA-ORA protocol by setting $K=1$ and thus leads to the worse performance when $K \geq 2$.
In the slotted ALOHA protocol, since both implementation details of the PHY coding\slash processing and the effects of fading are neglected in the protocol design phase, the PHY data rate of $\log_2(1+\textsf{snr})$ is adopted for fair comparison.
It is observed that the IA-ORA protocol is superior to the slotted ALOHA and ORA protocols in terms of the throughput in the low and moderate SNR regimes, and then converges to a certain value in the high SNR regime.
The throughput saturation for high SNR is because the MAC throughput and the multiuser diversity gain are not fully achieved due to the limited $N$ in the network suffering from the severe inter-cell interference, where $N$ is set to $100$ and is less than the value computed by our user scaling law in~(\ref{eq_condition}).
In other words, if $N$ is set to a higher value that scales with \textsf{snr}, then the resulting throughput achieved by our IA-ORA protocol will not experience this apparent saturation in the high SNR regime, which implies that the IA-ORA protocol will be still dominant.

\begin{figure}
\begin {center}
\epsfig{file=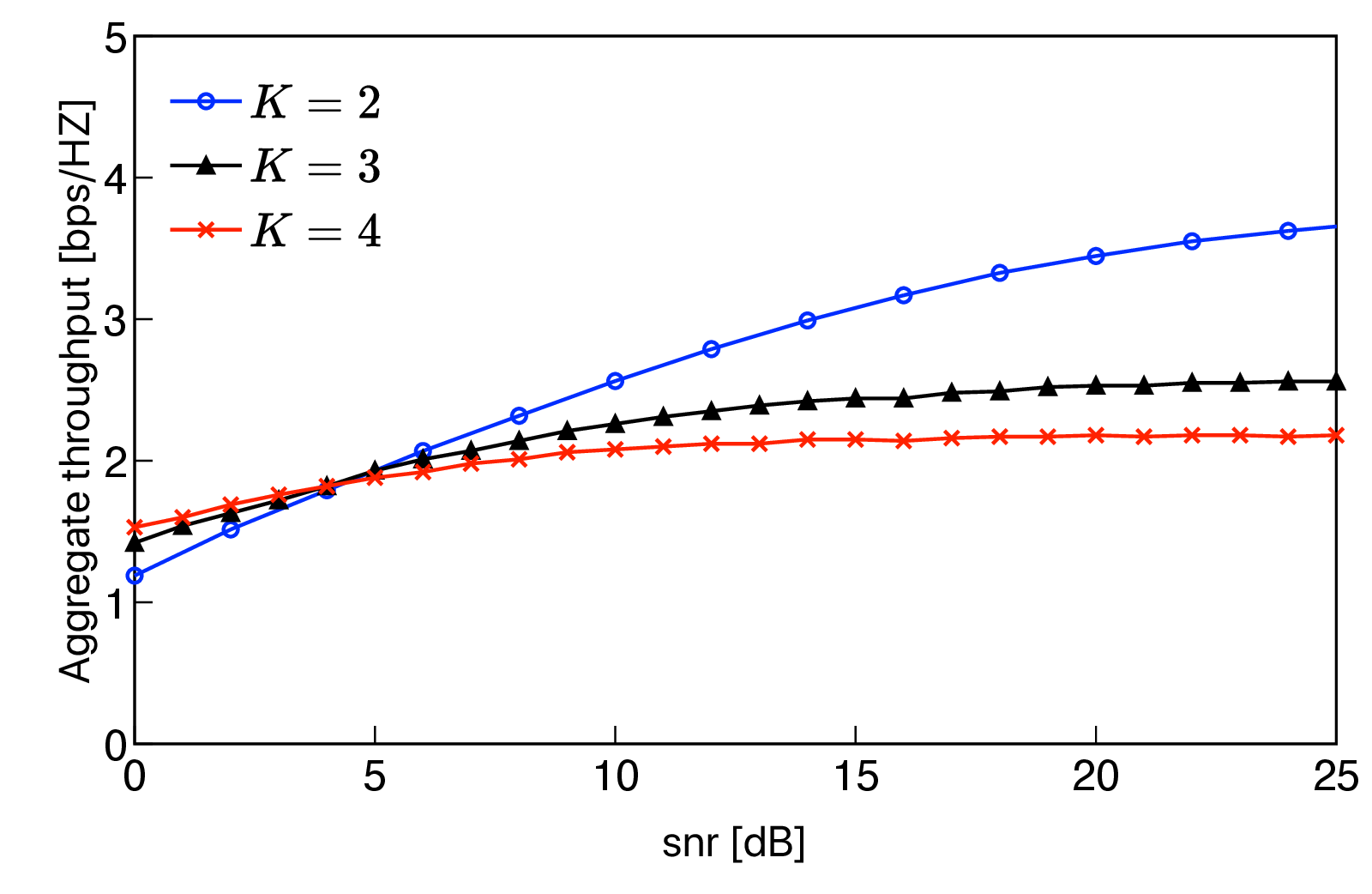, width=0.48\hsize}
\end {center}
\caption{The aggregate throughput versus $\textsf{snr}$, where the proposed IA-ORA protocol is employed for $K=\{2,3,4\}$ and $N=100$.}
\label{fig_Rsum_N100_K_sim}
\end{figure}

In Fig. \ref{fig_Rsum_N100_K_sim}, the aggregate throughput of our IA-ORA protocol versus $\textsf{snr}$ in dB scale is illustrated for $K = \{2, 3, 4\}$ and $N = 100$.
It is seen that the case of a larger $K$ leads to better performance on the aggregate throughput in the low SNR regime, but results in earlier throughput satuation with increasing $\textsf{snr}$ due to a \emph{more strict user scaling law} in Theorem 1, indicating that $N$ needs to increase exponentially with $K$ for given $\textsf{snr}$.
In other words, the curve for $K=2$ achieves inferior performance to the other curves for $K = \{3,4\}$ in the low SNR regime, whereas it tends to increase steadily with $\textsf{snr}$ and then becomes saturated at a relatively high $\textsf{snr}$ point.
Thus, for the plot showing the aggregate throughput of the IA-ORA and ORA protocols versus $\textsf{snr}$, it is also of importance to investigate a crossover at which two curves meet.
The crossover $\textsf{snr}$ and the resulting aggregate throughput $R_\text{sum}$ (i.e., ($R_\text{sum}$, $\textsf{snr}$)) are summarized in Table \ref{table_crossover} according to various settings of $K$ and $N$. From the table, it is observed that i) the crossover SNR point tends to decrease (i.e., the IA-ORA protocol becomes dominant in the relatively low SNR regime) with increasing $K$ for given $N$ since the IA-ORA protocol experiences more inter-cell interference that has not been sufficiently suppressed due to the limited $N$ and ii) the crossover SNR point tends to rise steadily with increasing $N$ for given $K$ since the inter-cell interference can be mitigated by virtue of the multiuser diversity gain, which implies that the performance of the IA-ORA protocol is superior to that of the ORA protocol even in the high SNR regime. Based on these findings, we are capable of designing a mode-switching strategy between the IA-ORA and ORA protocols depending on the values of $K$ and $N$, which results in consistently guaranteeing the best performance.\footnote{Note that approximate values of $K$ and $N$ can be estimated based on the average number of users (e.g., residents or visitors) measured over time in a given network environment.}

  \begin{table}[!t]
  \centering 
  \caption{The pair of ($R_\text{sum}$, $\textsf{snr}$) according to various settings of $K$ and $N$, where $\textsf{snr}$ denotes a crossover where two curves coincide.} \label{table_cross}
  \begin{tabular}{c | c | c | c | c}
  \hline
  $K \backslash N$ & 50 & 100 & 150 & 200\\
  \hline
  \hline
   2 & (2.83, 17dB) & (3.59, 22dB) & (4.05, 26dB) & (4.31, 28dB)\\
   3 & (1.91, 9dB) & (2.35, 12dB) & (2.68, 14dB) & (2.88, 16dB)\\
   4 & (1.69, 7dB) & (2.08, 10dB) & (2.28, 11dB) & (2.45, 12dB)\\
  \hline
  \end{tabular}
  \label{table_crossover}
\end{table}

\subsection{Performance Evaluation Under the Assumption of Imperfect Partial CAIT}
The numerical results in the previous subsections were obtained based on the assumption of perfect partial CAIT.
However, it may be hardly possible to acquire the perfect partial CAIT due to either the channel estimation or feedback error in practice.
To investigate the robustness of our IA-ORA protocol in the presence of channel uncertainty, we perform simulations by assuming the \emph{imperfect} channel gain $\hat{g}_{j \rightarrow k}^i = |\hat{h}_{j \rightarrow k}^i|^2$ at each transmitter, where $\hat{h}_{j \rightarrow k}^i = h_{j \rightarrow k}^i + \Delta h_{j \rightarrow k}^i$.
Here, the error component, $\Delta h_{j \rightarrow k}^i$, is modeled as an i.i.d. and complex Gaussian random variable with zero-mean and the variance $\sigma^2$\cite{zhang2008statistically}.
In the IA-ORA protocol, we use the optimal parameters $\Phi_G^\ast$ and $R^\ast$ found based on (\ref{eq_prac_opt}).

As illustrated in Fig. \ref{fig_Rsum_MSE_N100_SNR10}, it is obvious to see that the aggregate throughput of both IA-ORA and ORA protocols gets degraded as $\sigma^2$ increases.
It is also observed that in the presence of the error component, our IA-ORA protocol still outperforms the ORA protocol and achieves acceptable performance if the error variance $\sigma^2$ is below a tolerable level that leads to negligible performance degradation (e.g., $\sigma^2 < 10^{-2}$).
Moreover, when $\sigma^2$ is large (e.g., $\sigma^2 > 10^{-1}$), the performance gap between the IA-ORA and ORA protocols tends to be wider, showing that our IA-ORA protocol is less sensitive to the channel gain inaccuracy.
\begin{figure}
\begin {center}
\epsfig{file=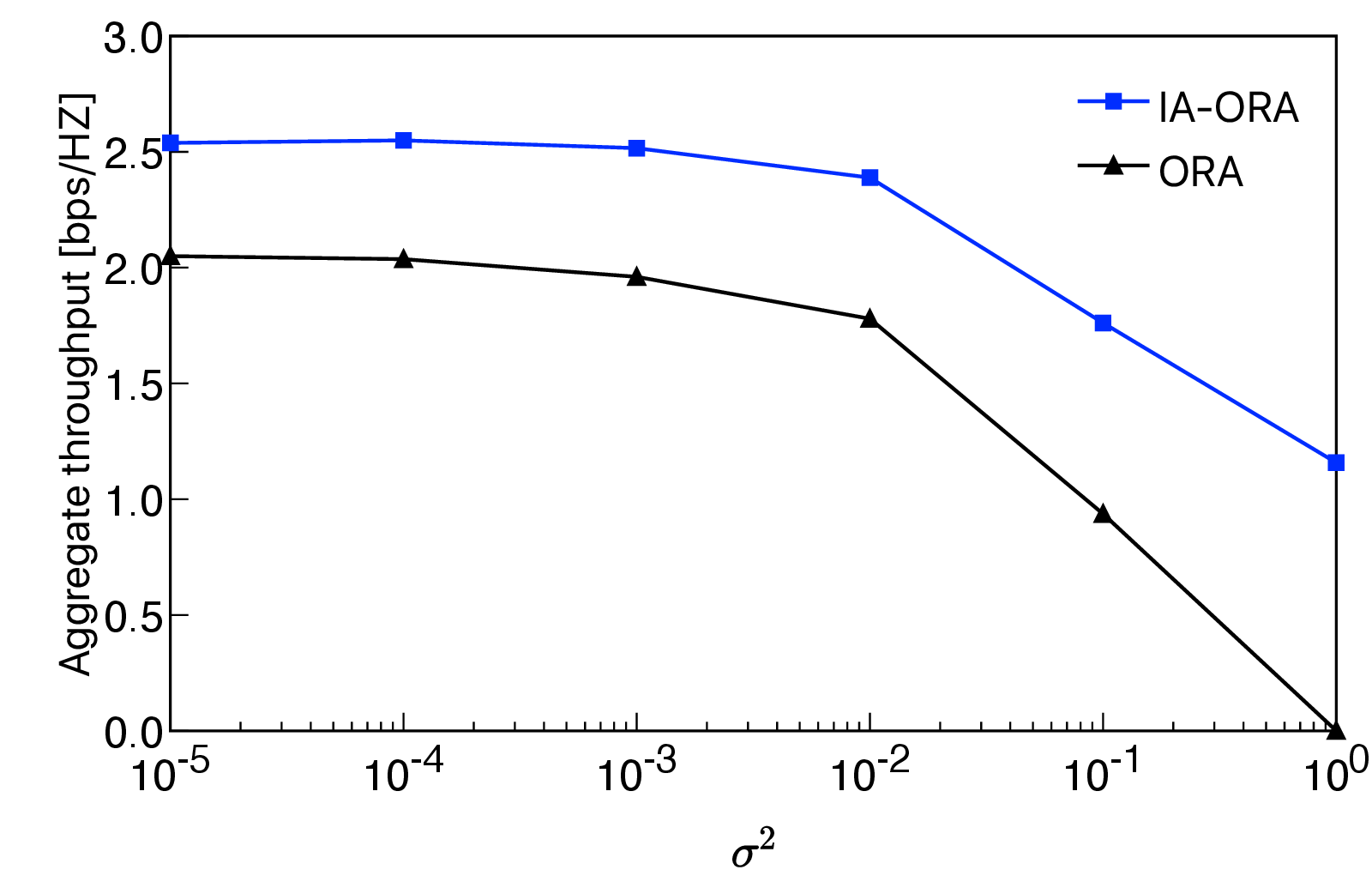, width=0.48\hsize}
\end {center}
\caption{The aggregate throughput versus $\sigma^2$ under the imperfect partial CAIT assumption, where the IA-ORA protocol and the conventional ORA protocol are employed for $\textsf{snr}$ $=10$dB, $K=2$, and $N=100$.}
\label{fig_Rsum_MSE_N100_SNR10}
\end{figure}

\section{Application of Opportunistic Random Access to CSMA/CA} \label{sec_CSMA/CA}

In this section, we extend our ORA protocol to a CSMA/CA random access network.

\subsection{System Model}
In conventional CSMS/CA networks, a user performs channel sensing and random BO before accessing the channel.
Each user selects a BO value $T_b$ uniformly at random from $\{0, 1, \cdots, CW\}$, where $CW$ is the contention window size with a mini-slot unit.
The user initially sets its BO timer to $T_b$ and keeps sensing the channel.
If the channel is busy, then the BO timer is frozen until the channel becomes idle.
Otherwise (i.e., on idle mode), the BO value is decreased by one at each (idle) mini-slot.
When the BO timer of a user expires (i.e., the BO value reaches zero), the user starts  transmission after a duration of one distributed inter-frame space (DIFS), referred to as the waiting time, that may consist of several mini-slots and a short inter-frame space (SIFS).
We assume that a time slot for data packet transmission is much longer than a mini-slot, corresponding to the unit in which channel sensing is performed.
Then, the channel gains are kept the same as long as a user performs channel sensing and random BO.

In CSMA/CA, a user can access the channel if it senses that the channel is idle and the corresponding BO timer expires.
The CSMA/CA protocol is designed in such a way that {\em only one user} transmits at a time and other users who sense the busy channel would avoid data transmission.
Unlike the case of our IA-ORA in slotted ALOHA, multi-packet transmission would not be possible under the CSMA/CA scenario.
Thus, since there is no interference at each receiver, the received signal $y_k \in \mathbb{C}$ at the $k$th AP is given by
\begin{equation}
\label{eq_rx_signal_csma}
y_k = h_{k \rightarrow k}^i x_k^i + z_k,
\end{equation}
where $x_k^i$ is the transmitted signal from the $i$th user in the $k$th cell.

\subsection{Overall Procedure}
We now describe opportunistic transmission in the CSMA/CA random access network. More specifically, we present a CSMA/CA protocol with opportunistic BO such that the BO timer of the user with higher channel gain $g_{k \rightarrow k}^i = |h_{k \rightarrow k}^i|^2$ is accelerated.
To this end, we introduce an acceleration factor $F_a$ for the BO value decrement, which is expressed as
\begin{equation}
\label{eq_fa}
F_{a} = \min ( {\lfloor} g_{k \rightarrow k}^i {\rfloor}, CW),
\end{equation}
where $\min(\cdot,\cdot)$ is imposed since the decrement step should not be larger than the contention window size $CW$.
For the user having a bad channel status (e.g., $g_{k\rightarrow k}^i \ll 1$), it follows that $F_a=0$, which implies that the BO timer pauses and the user would not attempt to access the channel.
In contrast, the user having a good channel status would access the channel faster.
For instance, suppose that $CW=8$ and there is a certain user with $2\le g_{k\rightarrow k}^i <3$ who randomly selects $T_b$ as 6.
Then, the user only needs to wait for 3 idle mini-slots to access the channel due to $F_a=2$.

\subsection{Numerical Evaluation}
\begin{figure}
\begin {center}
\epsfig{file=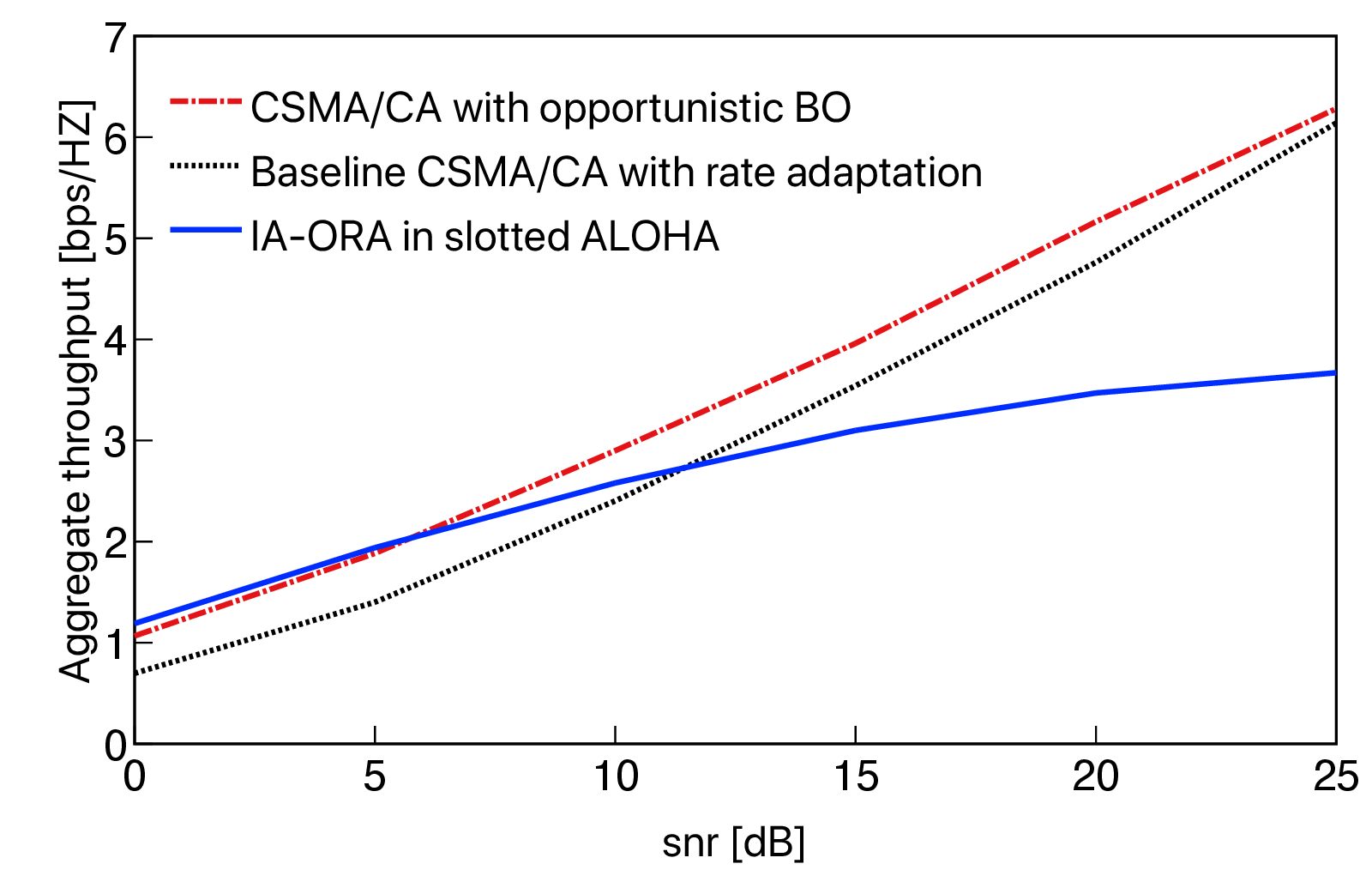, width=0.48\hsize}
\end {center}
\caption{The aggregate throughput versus $\textsf{snr}$, where the proposed  CSMA/CA protocol with opportunistic BO and the IA-ORA protocol in slotted ALOHA, as well as CSMA/CA protocol with rate adaption are employed for $K=2$ and $N=100$.}
\label{fig_Rsum_N100_K2_CSMA}
\end{figure}

We evaluate the aggregate throughput [bps/Hz] of the proposed CSMA/CA protocol with opportunistic BO through computer simulations.
In our CSMA/CA random access UDN scenario, we use the following parameter settings: the mini-slot duration $t_{sl} = 9\mu$s; the symbol duration  $t_{os} = 4\mu$s; the preamble duration $t_{pr} = 20\mu$s; the SIFS duration $t_{SIFS} = 16\mu$s; the DIFS duration $t_{DIFS} = t_{SIFS} + 2t_{sl} = 34\mu$s; the acknowledgement duration $t_{ACK} = t_{pr} + 5t_{os} = 40\mu$s; the number of symbols per packet is $10^4$,  and $CW = 1,023$.
As benchmark schemes, performance of a baseline CSMA/CA protocol with rate adaptation and the IA-ORA protocol used in slotted ALOHA is evaluated for comparison. In contrast to the IA-ORA case, each user is capable of adopting rate adaptation under the CSMA/CA protocol since there is no interference at the receiver; the $i$th user in the $j$th cell transmits with the PHY data rate $R_{k,i}=\log_2(1+g_{k\rightarrow k}^i\textsf{snr})$.

In Fig. \ref{fig_Rsum_N100_K2_CSMA}, the aggregate throughput of the proposed CSMA/CA with opportunistic BO and the two benchmark schemes versus \textsf{snr} in dB scale is plotted for $K=2$ and $N=100$. From the figure, the following observations are found: 1) the IA-ORA protocol is superior in the low SNR regime due to the multi-packet reception capability in IA-ORA; 2) in the high SNR regime, the two CSMA/CA methods outperform the IA-ORA protocol since collision reduction plays a more crucial role in the performance at high SNR; and 3) CSMA/CA with opportunistic BO offers promising gains over the baseline CSMA/CA consistently due to the fact that a waiting time of users with a high channel gain can be shortened. These findings reveal that there is a meaningful trade-off between the potential $K$-fold MAC throughput gain by IA-ORA and the interference (collision) avoidance by CSMA/CA.

\section{Concluding Remarks} \label{sec_conclusion}
In this paper, the decentralized IA-ORA protocol operating with partial CAIT was presented in the ultra-dense $K$-cell slotted ALOHA random access network, where no centralized coordination from the serving APs is required and thus the proposed protocol is appropriate for MTC or IoT. The aggregate throughput scaling achieved by the proposed protocol was then analyzed.
As our main analytical contribution, it was proved that the IA-ORA protocol asymptotically achieves the aggregate throughput scaling of $\frac{K}{e} (1-\epsilon) \log (\textsf{snr} \log N)$ in our multi-cell random access network if $N$ scales faster than $\textsf{snr}^{\frac{K-1}{1-\delta}}$ with respect to \textsf{snr} for small constants $\epsilon>0$ and $0<\delta<1$. To validate the analytical result and the effectiveness of our IA-ORA protocol, extensive computer simulations were also conducted---the throughput scaling and user scaling laws were confirmed numerically; the superiority of our protocol over benchmark schemes was shown in practical settings; and the robustness of our protocol was investigated when imperfect partial CAIT is assumed. Finally, our ORA protocol was extended to a CSMA/CA random access network, and then the trade-off between the $K$-fold MAC throughput gain by IA-ORA and the collision avoidance by CSMA/CA was discussed. 
Our random access framework would shed important insights for intelligently solving intra-cell collision and inter-cell interference problems.

Future research directions include extensions to networks with multiple antennas, opportunistic random access with user fairness, and cooperative slotted ALOHA systems. A potential avenue of future research in this area also includes the exploitation of interference awareness in CSMA/CA networks.

\appendices
\section{Proof of Lemma 1} \label{proof_lemma1}
\begin{proof}
[Proof: \nopunct] As shown in (\ref{eq_discrete_R}), the PHY data rate $R$ is a function of the integer $\nu \in \{0,1, \cdots, (K-1)N\}$.
We recall that the term $\Pr\left( \frac{\Phi_G}{\textsf{snr}^{-1} + i \Phi_I} \geq 2^R - 1 \right)$ in (\ref{eq_qs}) is given by $1$ and $0$ for $i \in \{ 0, \cdots, \nu \}$ and $i \in \{ \nu+1 \cdots, (K-1)N \}$, respectively, if $R$ is set to a value such that
\begin{equation} \nonumber
\label{ }
\frac{\Phi_G}{\textsf{snr}^{-1} + (\nu+1)\Phi_I} < 2^R-1 \leq \frac{\Phi_G}{\textsf{snr}^{-1} + \nu \Phi_I}.
\end{equation}
Thus, by partitioning the whole range of $i$ into two sub-ranges, (\ref{eq_qs}) can be expressed as
\begin{equation} \nonumber
\begin{split}
\tilde{p}_s &= \sum_{i=0}^{\nu} {1 \cdot \Pr(\tilde{n} = i)} + \sum_{i=\nu+1}^{(K-1)N} {0 \cdot \Pr(\tilde{n} = i)}\\
&= \sum_{i=0}^{\nu} {\Pr(\tilde{n} = i)}\\
&= \Pr( \tilde{n} \leq \nu ),
\end{split}
\end{equation}
which represents the CDF of the binomial random variable $\tilde{n} \sim B((K-1)N, \frac{1}{N})$. Since this CDF can be expressed as the regularized incomplete beta function \cite{abramowitz1966handbook}, it follows that
\begin{equation} \nonumber
\label{ }
\tilde{p}_s = I_{1-\frac{1}{N}} \left( (K-1)N-\nu, \nu+1 \right).
\end{equation}
This completes the proof of this lemma.
\end{proof}

\section{Proof of Theorem 1} \label{proof_theorem1}
\begin{proof}[Proof: \nopunct]
Using $\Phi_I = \textsf{snr}^{-1}$, (\ref{eq_bound_FI2}) can be rewritten as
\begin{equation}
\label{eq_ana_bound}
\begin{split}
R_\text{sum} \geq &K \left(1-\frac{1}{N}\right)^{N-1} \cdot \log_2\left( 1 + \frac{ \textsf{snr} \cdot \ln(F_{\tilde{I}}(\textsf{snr}^{-1}) N) }{ \nu + 1}  \right)\\
&\cdot I_{1-\frac{1}{N}} \left( (K-1)N-\nu, \nu+1 \right).
\end{split}
\end{equation}
We then select a finite value of $\nu$ that makes the function $I_{1-\frac{1}{N}} \left( (K-1)N-\nu, \nu+1 \right)$ approach almost one by following the equation below:
\begin{equation}
\label{eq_Q_star}
I_{1-\frac{1}{N}} \left((K-1)N-\nu^\ast, \nu^\ast+1 \right) = 1 - \epsilon,
\end{equation}
where $\epsilon>0$ is an arbitrarily small constant.\footnote{Note that $\epsilon$ is a given design parameter and $\nu^*$ can be expressed as a function of given $\epsilon$ and $K$.}
Since the function $I_{1-\frac{1}{N}} \left( (K-1)N-\nu, \nu+1 \right)$ is monotonically increasing with $\nu$ as shown in Fig. \ref{fig_Q}, we can find finite $\nu^*(K, \epsilon)$ by taking the inverse of the regularized incomplete beta function in (\ref{eq_Q_star}).

By substituting (\ref{eq_Q_star}) and $\nu = \nu^*(K, \epsilon)$ into (\ref{eq_ana_bound}), we have
\begin{equation} \nonumber
\label{eq_lb_final}
\begin{split}
R_\text{sum} \geq &K \left(1-\frac{1}{N}\right)^{N-1} \cdot \log_2\left( 1 + \frac{ \textsf{snr} \cdot \ln(F_{\tilde{I}}( \textsf{snr}^{-1}) N) }{ \nu^*(K, \epsilon) + 1}  \right)\\ &\cdot (1-\epsilon).
\end{split}
\end{equation}

From the fact that $\left(1-\frac{1}{N}\right)^{N-1}$ is monotonically decreasing with increasing $N$ and $\lim\limits_{N \rightarrow \infty}{\left(1-\frac{1}{N}\right)^{N-1}} = \frac{1}{e}$, the aggregate throughput is lower-bounded by
\begin{equation}  \nonumber
\label{eq_lb_final}
\begin{split}
R_\text{sum} &\geq \frac{K}{e} \cdot \log_2\left( 1 + \frac{ \textsf{snr} \cdot \ln(F_{\tilde{I}}( \textsf{snr}^{-1}) N) }{ \nu^*(K, \epsilon) + 1}  \right) \cdot (1-\epsilon)\\
&\geq \frac{K}{e} \cdot \log_2\left( 1 + \frac{ \textsf{snr} \cdot \ln( c_1 \left( \textsf{snr}^{-1} \right)^{K-1} N) }{ \nu^*(K, \epsilon) + 1}  \right) \cdot (1-\epsilon),
\end{split}
\end{equation}
where the second inequality follows due to Lemma \ref{lemma_FI_bound}.
In order to obtain the logarithmic gain (i.e., the power gain with increasing $N$), it should be fulfilled that $c_1 \left( \textsf{snr}^{-1} \right)^{K-1} N \geq N^\delta$ for $0<\delta<1$, which finally leads to
\begin{equation}
\label{eq_USC_theorem}
N = \Omega \left( \textsf{snr}^{\frac{K-1}{1-\delta}} \right).
\end{equation}
In consequence, under the user scaling condition in (\ref{eq_USC_theorem}), the aggregate throughput can be lower-bounded by
\begin{equation}
\label{eq_finalBound}
\begin{split}
R_\text{sum} \geq \frac{K}{e} (1-\epsilon) \cdot \log_2\left( 1 + \frac{ \delta \cdot  \textsf{snr} \ln( N) }{ \nu^*(K, \epsilon) + 1}  \right),
\end{split}
\end{equation}
which scales as $\frac{K}{e} (1-\epsilon) \log (\textsf{snr}\log N)$ with respect to \textsf{snr} from the fact that $\delta$ and $\nu^*(K, \epsilon)$ are some constants independent of $N$.
This completes the proof of the theorem.
\end{proof}

\section*{Acknowledgment}
The material in this paper has been presented in part at the IEEE International Conference on Communications, Paris, France, May 2017 \cite{lin2017achieving} and the IEEE International Symposium on Information Theory, Aachen, Germany, June 2017 \cite{lin2017multi}.


\begin{thebibliography}{00}

\bibitem{rajandekar2015survey}
A.~Rajandekar and B.~Sikdar, ``{A survey of MAC layer issues and
protocols for
  machine-to-machine communications},'' \emph{IEEE Internet Things J.}, vol.~2,
  no.~2, pp. 175--186, Jan. 2015.

\bibitem{Bockelmann2016}
C.~Bockelmann, N.~Pratas, H.~Nikopour, K.~Au, T.~Svensson,
C.~Stefanovic,
  P.~Popovski, and A.~Dekorsy, ``{Massive machine-type communications in 5G:
  physical and MAC-layer solutions},'' \emph{IEEE Commun. Mag.}, vol.~54,
  no.~9, pp. 59--65, Sep. 2016.

\bibitem{Bertsekas1992}
D.~Bertsekas and R.~Gallager, \emph{{Data networks (2nd
Ed.)}}.\hskip 1em plus
  0.5em minus 0.4em\relax New Jersey, USA: Prentice-Hall, Inc., 1992.

\bibitem{durisi2016toward}
G.~Durisi, T.~Koch, and P.~Popovski, ``Toward massive,
ultrareliable, and
  low-latency wireless communication with short packets,'' \emph{Proc. IEEE},
  vol. 104, no.~9, pp. 1711--1726, Sep. 2016.

\bibitem{paolini2015coded}
E.~Paolini, G.~Liva, and M.~Chiani, ``{Coded slotted ALOHA: A
graph-based
  method for uncoordinated multiple access},'' \emph{IEEE Trans. Inf. Theory},
  vol.~61, no.~12, pp. 6815--6832, Dec. 2015.

\bibitem{jakovetic2015cooperative}
D.~Jakoveti{\'c}, D.~Bajovi{\'c}, D.~Vukobratovi{\'c}, and
V.~Crnojevi{\'c},
  ``{Cooperative slotted ALOHA for multi-base station systems},'' \emph{IEEE
  Trans. Commun.}, vol.~63, no.~4, pp. 1443--1456, Feb. 2015.

\bibitem{chen2017machine}
S.~Chen, R.~Ma, H.-H. Chen, H.~Zhang, W.~Meng, and J.~Liu,
``Machine-to-machine
  communications in ultra-dense networks--a survey,'' \emph{IEEE Commun.
  Surveys Tuts.}, May 2017.

\bibitem{wyner1994shannon}
A.~D. Wyner, ``{Shannon-theoretic approach to a Gaussian cellular
  multiple-access channel},'' \emph{IEEE Trans. Inf. Theory}, vol.~40, no.~6,
  pp. 1713--1727, Aug. 1994.

\bibitem{somekh2000shannon}
O.~Somekh and S.~Shamai, ``{Shannon-theoretic approach to a
Gaussian cellular
  multiple-access channel with fading},'' \emph{IEEE Trans. Inf. Theory},
  vol.~46, no.~4, pp. 1401--1425, Aug. 2002.

\bibitem{cadambe2008interference}
V.~R. Cadambe and S.~A. Jafar, ``Interference alignment and
degrees of freedom
  of the {$K$}-user interference channel,'' \emph{IEEE Trans. Inf. Theory},
  vol.~54, no.~8, pp. 3425--3441, Jul. 2008.

\bibitem{gou2010degrees}
T.~Gou and S.~A. Jafar, ``{Degrees of freedom of the $K$ user $M
\times N$ MIMO
  interference channel},'' \emph{IEEE Trans. Inf. Theory}, vol.~56, no.~12, pp.
  6040--6057, Nov. 2010.

\bibitem{cadambe2008degrees}
V.~R. Cadambe and S.~A. Jafar, ``{Degrees of freedom of wireless X
networks},''
  in \emph{Proc. IEEE Int. Symp. Inf. Theory (ISIT)}, Toronto, Canada, Jul.
  2008, pp. 1268--1272.

\bibitem{yang2017opportunistic}
H.~J. Yang, W.-Y. Shin, B.~C. Jung, C.~Suh, and A.~Paulraj,
``Opportunistic
  downlink interference alignment for multi-cell {MIMO} networks,'' \emph{IEEE
  Trans. Wireless Commun.}, vol.~16, no.~3, pp. 1533--1548, Mar. 2017.

\bibitem{suh2008interference}
C.~Suh and D.~Tse, ``Interference alignment for cellular
networks,'' in
  \emph{Proc. 46th Annu. Allerton Conf.}, Monticello, IL, Sep. 2008, pp.
  1037--1044.

\bibitem{jung2011opportunistic}
B.~C. Jung and W.-Y. Shin, ``{Opportunistic interference alignment
for
  interference-limited cellular TDD uplink},'' \emph{IEEE Commun. Lett},
  vol.~15, no.~2, pp. 148--150, Feb. 2011.

\bibitem{jung2012opportunistic}
B.~C. Jung, D.~Park, and W.-Y. Shin, ``Opportunistic interference
mitigation
  achieves optimal degrees-of-freedom in wireless multi-cell uplink networks,''
  \emph{IEEE Trans. Commun.}, vol.~60, no.~7, pp. 1935--1944, Jul. 2012.

\bibitem{gollakota2009interference}
S.~Gollakota, S.~D. Perli, and D.~Katabi, ``Interference alignment
and
  cancellation,'' \emph{ACM SIGCOMM Computer Communication Review}, vol.~39,
  no.~4, pp. 159--170, Aug. 2009.

\bibitem{lin2011random}
K.~C.-J. Lin, S.~Gollakota, and D.~Katabi, ``{Random access
heterogeneous MIMO
  networks},'' \emph{ACM SIGCOMM Computer Communication Review}, vol.~41,
  no.~4, pp. 146--157, Oct. 2011.

\bibitem{jin2015opportunistic}
H.~Jin, S.-W. Jeon, and B.~C. Jung, ``Opportunistic interference
alignment for
  random access networks,'' \emph{IEEE Trans. Veh. Technol.}, vol.~64, no.~12,
  pp. 5947--5954, Jan. 2015.

\bibitem{zhou2014bbn}
W.~Zhou, T.~Bansal, P.~Sinha, and K.~Srinivasan, ``{BBN:
Throughput scaling in
  dense enterprise WLANs with bind beamforming and nulling},'' in \emph{Proc.
  ACM 20th Annu. Int. Conf. on Mobile Computing and Networking}, Maui, HI, Sep.
  2014, pp. 165--176.

\bibitem{lin2015simple}
H.~Lin, K.~Ishibashi, W.-Y. Shin, and T.~Fujii, ``A simple random
access scheme
  with multilevel power allocation,'' \emph{IEEE Commun. Lett.}, vol.~19,
  no.~12, pp. 2118--2121, Dec. 2015.

\bibitem{lin2016decentralized}
------, ``Decentralized power allocation for secondary random access in
  cognitive radio networks with successive interference cancellation,'' in
  \emph{Proc. of IEEE Int. Conf. Commun. (ICC)}, Kuala Lumper, Malaysia, May
  2016, pp. 1--6.
  
  \bibitem{HZhao_TVT}
H.~Zhao, K.~Ding, N.~I.~Sarkar, J.~Wei, and J.~Xiong, ``A simple distributed channel allocation algorithm for {D2D} communication pairs,'' \emph{IEEE Trans. Veh. Technol.}, vol.~67, no.~11, pp. 10960--10969, Nov. 2018.

\bibitem{JLi_TMC}
J.~Li, H.~Zhao, J.~Wei,~and D.~Ma, ``Sender-jump receiver-wait: {A} simple blind rendezvous algorithm for distributed cognitive radio networks,'' \emph{IEEE Trans. Mob. Comput.}, vol.~17, no.~1, pp. 183--196, Jan. 2018.

\bibitem{JLi_IoTJ}
J.~Li, H.~Zhao, A.~S.~Hafid, J.~Wei, H.~Yin, and B.~Ren, ``A bio-inspired solution to cluster-based distributed spectrum allocation in high-density cognitive {Internet of Things},'' \emph{IEEE Internet Things J.}, vol.~6, no.~6, pp. 9294--9307, Dec. 2019.

 \bibitem{JLi_TCOM}
J.~Li, H.~Zhao, S.~Zhang, A.~S.~Hafid, D.~Niyato, and J.~Wei, ``Cross-layer analysis and optimization on access delay in channel-hopping-based distributed cognitive radio networks,'' \emph{IEEE Trans. Commun.}, vol.~67, no.~7, pp. 4654--4668, Jul. 2019.

 \bibitem{SZhang_TWC}
S.~Zhang, A.~S.~Hafid, H.~Zhao, S.~Wang, ``Cross-layer rethink on sensing-throughput tradeoff for multi-channel cognitive radio networks,'' \emph{IEEE Trans. Wireless Commun.}, vol.~15, no.~10, pp. 6883--6897, Oct. 2016.

 \bibitem{HZhao_JNCA}
H.~Zhao, J.~Wei, N.~I.~Sarkar, S.~Huang, ``{E-MAC}: {An} evolutionary solution for collision avoidance in wireless ad hoc networks,'' \emph{J. Netw. Comput. Appl.}, vol.~65, pp. 1--11, Apr. 2016.

\bibitem{knopp1995information}
R.~Knopp and P.~A. Humblet, ``Information capacity and power
control in
  single-cell multiuser communications,'' in \emph{Proc. IEEE Int. Conf.
  Commun. (ICC)}, Seattle, WA, Jun. 1995, pp. 331--335.

\bibitem{viswanath2002opportunistic}
P.~Viswanath, D.~N.~C. Tse, and R.~Laroia, ``Opportunistic
beamforming using
  dumb antennas,'' \emph{IEEE Trans. Inf. Theory}, vol.~48, no.~6, pp.
  1277--1294, Aug. 2002.

\bibitem{sharif2005capacity}
M.~Sharif and B.~Hassibi, ``{On the capacity of MIMO broadcast
channels with
  partial side information},'' \emph{IEEE Trans. Inf. Theory}, vol.~51, no.~2,
  pp. 506--522, Jan. 2005.
  
\bibitem{Shin2012network}
W.-Y. Shin and B. C. Jung, ``Network coordinated opportunistic beamforming in downlink cellular networks,'' \emph{IEICE Trans. Commun.}, vol.~E95--B,
  no.~4, pp. 1393--1396, Apr. 2012.

\bibitem{shin2012can}
W.-Y. Shin, D.~Park, and B.~C. Jung, ``Can one achieve multiuser
diversity in
  uplink multi-cell networks?'' \emph{IEEE Trans. Commun.}, vol.~60, no.~12,
  pp. 3535--3540, Dec. 2012.

\bibitem{qin2006distributed}
X.~Qin and R.~A. Berry, ``Distributed approaches for exploiting
multiuser
  diversity in wireless networks,'' \emph{IEEE Trans. Inf. Theory}, vol.~52,
  no.~2, pp. 392--413, Jan. 2006.

\bibitem{adireddy2005exploiting}
S.~Adireddy and L.~Tong, ``Exploiting decentralized channel state
information
  for random access,'' \emph{IEEE Trans. Inf. Theory}, vol.~51, no.~2, pp.
  537--561, Jan. 2005.

\bibitem{wang2009transmission}
S.-H. Wang and Y.-W.~P. Hong, ``{Transmission control with
imperfect CSI in
  channel-aware slotted ALOHA networks},'' \emph{IEEE Trans. Wireless Commun.},
  vol.~8, no.~10, pp. 5214--5224, Oct. 2009.

\bibitem{moon2016channel}
H.~Moon, ``Channel-adaptive random access with discontinuous
channel
  measurements,'' \emph{IEEE J. Sel. Areas Commun.}, vol.~34, no.~5, pp.
  1704--1712, Apr. 2016.

\bibitem{miao2012channel}
G.~Miao, Y.~Li, and A.~Swami, ``Channel-aware distributed medium
access
  control,'' \emph{IEEE/ACM Trans. Netw.}, vol.~20, no.~4, pp. 1290--1303, Dec.
  2012.

\bibitem{Khanian2016}
Z.~Khanian, M.~Rasti, F.~Salek, and E.~Hossain, ``A distributed
opportunistic
  {MAC} protocol for multichannel wireless networks,'' \emph{IEEE Trans.
  Wireless Commun.}, vol.~15, no.~6, pp. 1--1, Mar. 2016.

\bibitem{liu2015stay}
Y.~Liu and M.~Liu, ``{To stay or to switch: Multiuser
multi-channel dynamic
  access},'' \emph{IEEE Trans. Mobile Comput.}, vol.~14, no.~4, pp. 858--871,
  Apr. 2015.

\bibitem{li2016maximum}
Y.~Li and L.~Dai, ``Maximum sum rate of slotted {Aloha} with
capture,''
  \emph{IEEE Trans. Commun.}, vol.~64, no.~2, pp. 690--705, Feb. 2016.

\bibitem{zhang2009does}
Y.~J. Zhang, P.~X. Zheng, and S.~C. Liew, ``How does
multiple-packet reception
  capability scale the performance of wireless local area networks?''
  \emph{IEEE Trans. Mobile Comput.}, vol.~8, no.~7, pp. 923--935, Jul. 2009.

\bibitem{bae2014achieving}
Y.~H. Bae, B.~D. Choi, and A.~S. Alfa, ``Achieving maximum
throughput in random
  access protocols with multipacket reception,'' \emph{IEEE Trans. Mobile
  Comput.}, vol.~13, no.~3, pp. 497--511, Mar. 2014.

\bibitem{komulainen2013effective}
P.~Komulainen, A.~T{\"o}lli, and M.~Juntti, ``{Effective CSI
signaling and
  decentralized beam coordination in TDD multi-cell MIMO systems},'' \emph{IEEE
  Trans. Sig. Process.}, vol.~61, no.~9, pp. 2204--2218, May 2013.

\bibitem{rogalin2014scalable}
R.~Rogalin, O.~Y. Bursalioglu, H.~Papadopoulos, G.~Caire, A.~F.
Molisch,
  A.~Michaloliakos, V.~Balan, and K.~Psounis, ``{Scalable synchronization and
  reciprocity calibration for distributed multiuser MIMO},'' \emph{IEEE Trans.
  Wireless Commun.}, vol.~13, no.~4, pp. 1815--1831, Apr. 2014.

\bibitem{IEEE802.11ac}
``{Local and metropolitan area networks specific requirements;
part 11:
  wireless LAN medium access control (MAC) and physical layer (PHY)
  specifications; amendment 4: enhancements for very high throughput for
  operation in bands below 6 GHz},'' \emph{IEEE P802.11ac/D3.1}, 2012.

\bibitem{liao2014mu}
R.~Liao, B.~Bellalta, M.~Oliver, and Z.~Niu, ``{MU-MIMO MAC
protocols for
  wireless local area networks: A survey},'' \emph{IEEE Commun. Surveys Tuts.},
  vol.~18, no.~1, pp. 162--183, 2014.

\bibitem{abramowitz1966handbook}
M.~Abramowitz and I.~A. Stegun, \emph{{Handbook of Mathematical
  Functions}}.\hskip 1em plus 0.5em minus 0.4em\relax New York, NY, USA: Dover,
  1966.

\bibitem{zhang2008statistically}
X.~Zhang, D.~P. Palomar, and B.~Ottersten, ``{Statistically robust
design of
  linear MIMO transceivers},'' \emph{IEEE Trans. Sig. Process.}, vol.~56,
  no.~8, pp. 3678--3689, Aug. 2008.

\bibitem{lin2017achieving}
H.~Lin and W.-Y. Shin, ``Achieving the optimal throughput scaling
in multi-cell
  random access networks,'' in \emph{Proc. of IEEE Int. Conf. Commun. (ICC)},
  Paris, France, May 2017.

\bibitem{lin2017multi}
------, ``Multi-cell aware opportunistic random access,'' in \emph{Proc. of
  IEEE Int. Symp. Inf. Theory (ISIT)}, Aachen, Germany, Jun. 2017.

\end{thebibliography}
\end{document}